\pdfoutput=1
\documentclass[journal,10pt]{IEEEtran}

\def\BigRoman{\uppercase\expandafter{\romannumeral\number\count 255 }}
\def\Romannumeral{\afterassignment\BigRoman\count255=} \makeatletter \newcommand*{\rom}[1]{\expandafter\@slowromancap\romannumeral #1@} \makeatother
\usepackage{amsmath,amssymb,float,amsthm,balance}

\newtheorem{Lemma}{Lemma}
\usepackage{color}
\usepackage{multirow}
\newtheorem{proof pf Lemma 6}{Proof}

\newtheorem{Property}{Property}
\usepackage{epstopdf}
\usepackage{enumitem}
\usepackage{diagbox}
\usepackage{mnsymbol}

\usepackage{algorithm}

\usepackage{cite}

\ifCLASSINFOpdf
\usepackage[pdftex]{graphicx}

\graphicspath{{../pdf/}{../jpeg/}}

\else

\fi

\usepackage{amsmath}

\usepackage{algorithmic}

\usepackage{array}

\usepackage[caption=false,font=footnotesize]{subfig}

\usepackage[bottom]{footmisc}

\usepackage{fixltx2e}

\ifCLASSOPTIONcaptionsoff
\usepackage[nomarkers]{endfloat}
\let\MYoriglatexcaption\caption
\renewcommand{\caption}[2][\relax]{\MYoriglatexcaption[#2]{#2}}
\fi

\newcommand{\upperRomannumeral}[1]{\uppercase\expandafter{\romannumeral#1}}

\theoremstyle{definition}

\DeclareMathOperator*{\argmax}{arg\,max}

\begin{document}

\title{Design and Analysis of LoS MIMO Systems with Uniform Circular Arrays}

\author{Yuri Jeon, Gye-Tae Gil, ~\IEEEmembership{Member,~IEEE}, and Yong H. Lee, ~\IEEEmembership{Senior Member,~IEEE}
\thanks{This work has been submitted to the IEEE for possible publication. Copyright may be transferred without notice, after which this version may no longer be accessible.}
\thanks{Y. Jeon and Y. H. Lee are with the School of Electrical Engineering, Korea
	Advanced Institute of Science and Technology, Daejeon 34141, Korea (e-mail:
	yuri0703@kaist.ac.kr; yohlee@kaist.ac.kr).}
\thanks{G. T. Gil  is with the Institute for Information Technology Convergence,
	Korea Advanced Institute of Science and Technology, Daejeon 305-701,
	South Korea (e-mail: gategil@kaist.ac.kr).}
}
{}

\maketitle

\begin{abstract}
	We consider the design of a uniform circular array (UCA) based multiple-input multiple-output (MIMO) system over line-of-sight (LoS) environments in which array misalignment exists. In particular, optimal antenna placement in UCAs and  transceiver architectures to achieve the maximum channel capacity without the knowledge of misalignment components are presented. To this end, we first derive a generic channel model of UCA-based LoS MIMO systems in which three misalignment factors including relative array rotation, tilting and center-shift are reflected concurrently. By factorizing the channel matrix into the singular value decomposition (SVD) form, we demonstrate that the singular values of UCA-based LoS MIMO systems are \textit{independent} of tilting and center-shift. Rather, they can be expressed as a function of the \textit{radii product-to-distance ratio} (RPDR) and the angle of relative array rotation. Numerical analyses of singular values show that the RPDR is a key design parameter of UCA systems. Based on this result, we propose an optimal design method for UCA systems which performs a one-dimensional search of RPDR to maximize channel capacity. It is observed that the channel matrix of the optimally designed UCA system  is close to an orthogonal matrix; this fact allows channel capacity to be achieved by a simple zero-forcing (ZF) receiver. Additionally, we propose a low-complexity precoding scheme for UCA systems in which the optimal design criteria cannot be fulfilled because of limits on array size.  The simulation results demonstrate the validity of the proposed design method and transceiver architectures.
\end{abstract}


\IEEEpeerreviewmaketitle

\section{Introduction}

\IEEEPARstart{D}{ue} to the potential in fixed wireless applications such as cellular backhaul, line-of-sight (LoS) multiple-input multiple-output (MIMO) communication systems with fixed transmitter and receiver locations have been proposed [1]\textendash[12]. In these systems, the LoS MIMO channel is nearly stationary and can be assumed to be deterministic. To provide the high data rates required for such applications, LoS MIMO systems exploit spatial multiplexing gains attained by appropriate spacing between antenna elements.
It has been shown that LoS MIMO systems can achieve a full multiplexing gain through the optimization of antenna placement [2]\textendash[9]. In particular, the orthogonality conditions that make the columns of  LoS MIMO channel matrices orthogonal, are derived in terms of the carrier wavelength, array size, and communication range. 
Furthermore, sensitivity to deviations caused by small misalignments, such as rotations, tilts, or translations has been analyzed [4]\textendash[9]. These results have been derived mostly for uniform linear arrays (ULAs) [4]\textendash[7] and uniform rectangular arrays (URAs) [8]\textendash[9].

Recently, LoS MIMO communication systems with uniform-circular-arrays (UCAs) have been drawing attention due to their advantages over ULAs and URAs [10]\textendash[15]: a salient feature of a UCA-based LoS MIMO system is that its channel matrix can be modeled as a circulant matrix [10]\textendash[12]; thus, the channel can be diagonalized by employing a  discrete Fourier transform (DFT) precoder and an inverse DFT (IDFT) combiner. This property can simplify the design and implementation of UCA-based LoS MIMO systems, and it allows us to regard UCA-based LoS MIMO systems as a candidate scheme for realizing orbital-angular momentum (OAM) transmission [12], [16]\textendash[19]. However, these systems have not been fully investigated, and their use for practical applications is limited. The orthogonality condition for optimal antenna spacing has been derived only for UCAs with three or four antenna elements [11]. In addition, the asymptotic analysis in [10] indicates that a  UCA-based MIMO channel can hardly satisfy the orthogonality condition and has jagged singular values, where some of them are too small to be used for signal transmission. In other words, the condition number of the channel matrix can be large, and it is difficult to exploit the full multiplexing gain. To overcome these difficulties, an optimal design method which is a process of finding the optimal radii of UCAs to maximize channel capacity was presented [13]. However, all of the research efforts consider only a scenario in which the transmitter (Tx) and the receiver (Rx) UCAs are perfectly aligned to each other. 

While there has been a lack of discussion on the optimal design criteria of UCA systems under array misalignment, it has traditionally been  believed that misalignments on UCAs should be compensated to obtain high performance gain. In this regard, channel models of misaligned UCA systems and methods to compensate for those misalignments have been developed [14], [19]. Channel models of center-shifted and tilted UCA systems were presented in [14] and [19], respectively, but a generic channel model that considers all kinds of misalignment concurrently has not yet been developed. Moreover, the misalignment compensation methods of [14] and [19] leverage the phase deviation of channel gain caused by misalignments, but its estimation method, which is an important issue for realization, was not presented. In practice, it is challenging to estimate misalignment angles when they are small or vary rapidly. Therefore, in-depth analysis of sensitivity to misalignment and finding a way to implement UCA systems without  estimating or  compensating misalignment components  are required.

In this work,  we develop an optimal design method for a UCA system  under array misalignment and present transceiver architectures   that achieve the channel capacity without estimating or compensating misalignment components. Here, we first present a generalized channel model of a misaligned UCA system in which all kinds of misalignments including tilting, center-shift, and array rotation, are taken into account. Then, it is shown that the singular values of the misaligned UCA system are independent of tilting and center-shift, but can be expressed as a function of the \textit{radii product-to-distance ratio} (RPDR) and the angle of relative array rotation. Further numerical analyses verify that the singular values are robust to array rotation, but fluctuate with the RPDR. These findings indicate that the RPDR is the key parameter for designing UCA systems. Based on this observation, we developed an optimal design method for UCA systems that performs a one-dimensional search for the RPDR to maximize the channel capacity. It is observed that the channel matrix of the misaligned UCA system is close to an orthogonal matrix when the optimal design criterion is satisfied; therefore, the maximum channel capacity can be achieved by a  simple zero-forcing (ZF) receiver. In addition, we develop a low-complexity precoding scheme for a UCA system in which the optimal RPDR value cannot be fulfilled due to limits on array size. The proposed precoding scheme consists of approximated power allocations, which can be implemented only by the information of the communication distance, and a codebook-based precoding framework where the codebook is designed by quantized angles of center-shift to avoid estimation of misalignment angles. The results indicate that the proposed precoding scheme almost achieves  the channel capacity with a small feedback overhead; thus, it can be a useful alternative to an  optimal precoder that requires either estimation and feedback of misalignment angles or full channel-state information. 

The remainder of this paper is organized as follows. Section II presents the channel model of a misaligned UCA system in which three types of misalignments, i.e., rotation, tilting and center-shift, are taken into account. In Section III, singular value analyses are performed, and the optimal design method of the misaligned UCA system is  presented. The design of a precoder for non-optimal UCA systems is presented in Section IV, and simulation results demonstrating the validity of the proposed design method and transceiver architectures are presented in Section V. Finally, Section VI presents the conclusion.


$\textsl{\itshape Notations}$: Bold upper-case $\mathbf{A}$ denotes a matrix and bold lower-case $\mathbf{a}$ denotes a vector. Superscripts $\mathbf{A}^T$ and $\mathbf{A}^H$ denote the transpose and the conjugate transpose of a matrix $\mathbf{A}$, respectively, and  $a^*$ denotes the complex conjugate of a complex number $a$. The $(n,m)^{th}$ entry and the $k$th column of a matrix $\mathbf{A}$ are denoted as $a(n,m)$ and $\mathbf{A}(:, k)$, respectively, and \textit{diag}(${a}_1,\cdots,{a}_N$) indicates a diagonal matrix whose diagonal entries are given by $\{{a}_1,\cdots,{a}_N\}$. The matrix obtained by taking the magnitudes of the entries of $\mathbf{A}$ is denoted as $|\mathbf{A}|$ (the $(n,m)^{th}$ entry of $|\mathbf{A}|$ is equal to $|a(n,m)|$), and  
$\mathbf{I}_{N}$ denotes the $N-$dimensional identity matrix. The matrices $\mathbf{R}_{\theta}^{xy}$, $\mathbf{R}_{\theta}^{xz}$ and  $\mathbf{R}_{\theta}^{yz}$ are rotation matrices representing the rotation on the $xy-$plane about the $z-$axis, on the $xz-$plane about the $y-$axis, and the on $yz-$plane about the $x-$axis, respectively. For example, 
\begin {equation}
\mathbf{R}_{\theta}^{xz}=\begin{bmatrix}
	\cos{\theta} &  0 & -\sin{\theta}\\
	0 & 1 &  0\\
	\sin{\theta} & 0 & \cos{\theta}\\
\end{bmatrix}.\nonumber
\end{equation}


\section{System Model}

In this section, we present a system model for misaligned UCA systems. It is assumed, without loss of generality, that misalignments are caused by the misplacement of an Rx UCA, given a well-positioned Tx UCA. The misplacement is modelled by the rotation, tilting and center-shift of an Rx UCA. We shall consider these misplacement models one by one, and then present a general model that jointly considers all three types of misplacements. The misalignment model presented in this section is an extension of the model introduced in [14] that considers rotation and center-shift.

\begin{figure}[t]
\centering
\centerline{\includegraphics[width=0.23\textwidth, height=0.15\textwidth]{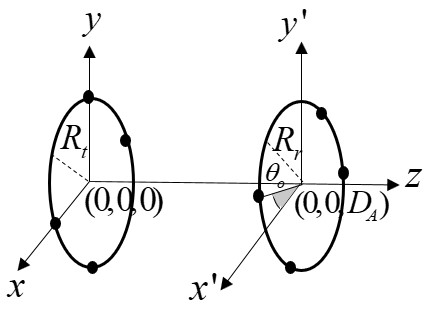}}
\caption{{Perfectly aligned UCA system when $N_s=4$. }}
\end{figure}

\subsection{Aligned UCAs with Rotation}
We consider a UCA system over an LoS channel that employs UCAs with $N_s$ antenna elements at the transmitter (Tx) and the receiver (Rx). The radii of the Tx and Rx UCAs are denoted as $R_t$ and $R_r$, respectively. For the aligned UCAs, we assume that the Tx UCA is located on the $xy-$plane and is centered at the coordinate $(x, y, z) = (0, 0, 0)$; the Rx UCA is located on the $x’y’-$plane, which is parallel to the $xy-$plane, and centered at $(x, y, z) = (0, 0, D_A)$, where $D_A \gg R_t$ and $R_r$. The Tx UCA is assumed to be fixed, and without loss of generality, the first Tx antenna element is located on the $x-$axis. The coordinates of the $m^{th}$ Tx antenna are given by $(R_t \cos\theta_m,R_t \sin\theta_m, 0)$ for $\theta_m = 2 \pi m/N_s$ and $m\in\{1,\cdots,N_s\}$. For the Rx UCA, we allow rotation by $\theta_o$ about the $z-$axis, where $-\pi/N_s\leq\theta_o\leq \pi/N_s$ (Fig. 1). The coordinates of the $n^{th}$ Rx antenna after rotation, denoted as $(x_{\theta_o},y_{\theta_o},z),$ are given by
\begin {equation}
\begin {split}
[x_{\theta_o},y_{\theta_o},z]^T&=\mathbf{R}_{\theta_o}^{xy}[R_r \cos\theta_n, R_r \sin\theta_n, D_A]^T\\&=[R_r\cos(\theta_n+\theta_o),R_r\sin(\theta_n+\theta_o),D_A]^T.
\end {split}
\end {equation}
where $\mathbf{R}_{\theta_o}^{xy}$ is the rotation matrix representing the rotation on the $xy-$plane about the $z-$axis, $\theta_n = 2 \pi n/N_s$, and $n\in\{1,\cdots,N_s\}$. The distance between the $m^{th}$ Tx antenna and the $n^{th}$ Rx antenna, denoted as $\tilde{d}_A(n, m)$, is given by
\begin {equation}
\tilde{d}_A(n,m) = \big\{D_A^2+R_t^2+R_r^2-2R_tR_r\cos(\theta_n-\theta_m+\theta_o)\big\}^{1 \over 2}.
\end {equation}
Because $D_A \gg R_t$ and $R_r$, ignoring the approximation errors, $\tilde{d}_A(n,m)$ can be rewritten as
\begin {equation}
\tilde{d}_A(n,m) = {D_A-{R_tR_r \over D_A}\cos(\theta_n-\theta_m+\theta_o)}.
\end {equation}

\begin{figure}[t]
	\centering
	\centerline{\includegraphics[width=0.15\textwidth, height=0.18\textwidth]{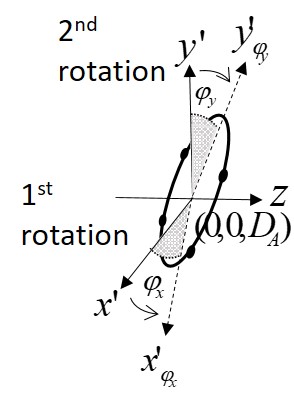}}
	\caption{{ Tilting modeled as a cascade of two rotations.}}
\end{figure}

\subsection{Modeling Rx UCA Tilting}

Referring to Fig. 2, the Rx UCA tilting can be represented as a cascade of two rotations: the UCA rotation about the $x'-$axis by the angle $\varphi_{y}$ followed by the rotation about the $y'_{\varphi_y}-$axis by the angle $\varphi_x$. The first rotation represents the tilt to the $x'z-$plane, and the second rotation represents the tilt to the $y'z-$plane. Because the $x'y'-$plane is parallel to the $xy-$ plane (Fig. 1), the coordinates of an antenna on the Rx UCA after the first rotation are given by $(x, y_{\varphi_y}, z_{\varphi_{y}})$, where ${[x, y_{\varphi_y}, z_{\varphi_y} ]}^T= \mathbf{R}_{\varphi_y}^{yz} [x, y, z]^T$. The coordinate after the second rotation is given by $(x_{\varphi_x}, y_{\varphi_y}, z_{\varphi_{x},\varphi_{y}})$, where ${[x_{\varphi_x},  y_{\varphi_y}, z_{\varphi_{x},\varphi_{y}} ]}^T=\mathbf{R}_{\varphi_x}^{xz} [x, y_{\varphi_y}, z_{\varphi_y}]^T$. Combining these results, the coordinate after the two types of rotations (or tilting) is given by
\begin {equation}
{[x_{\varphi_x}, y_{\varphi_y},  z_{\varphi_{x},\varphi_{y}} ]}^T=
\mathbf{R}_{\varphi_x}^{xz}\mathbf{R}_{\varphi_y}^{yz} [x, y, z]^T.
\end {equation}
When both the rotation considered in (1) and the tilting in (4) occur, the UCA coordinate is given by $(x_{{\theta_o},\varphi_x}, y_{{\theta_o},\varphi_y},  z_{\varphi_{x},\varphi_{y}})$ where 
\begin {equation}
{[x_{{\theta_o},\varphi_x}, y_{{\theta_o},\varphi_y},  z_{\varphi_{x},\varphi_{y}} ]}^T=
\mathbf{R}_{\varphi_x}^{xz}\mathbf{R}_{\varphi_y}^{yz}[x_{\theta_o},y_{\theta_o},z]^T
\end {equation}
for $(x_{\theta_o},y_{\theta_o})$ in (1).
\subsection{Modeling Rx UCA Center Shift}

\begin{figure}
	\centering
	\subfloat[]{\includegraphics[width=0.2\textwidth, height=0.14\textwidth]{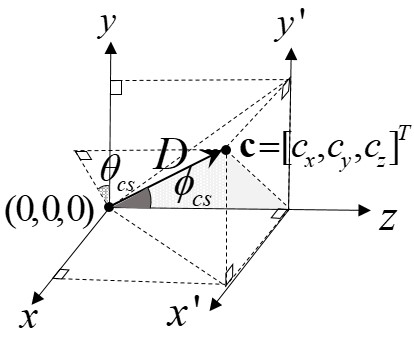}}\quad\quad
	\subfloat[]{\includegraphics[width=0.2\textwidth, height=0.13\textwidth]{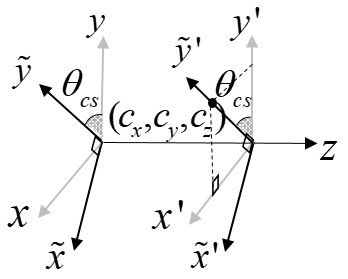}}
	\caption{(a) Modeling the center shift. Ideally, the center is supposed to be located at $(0,0, D_A)$, but it is shifted to $(c_x,c_y,c_z)$ due to misalignment. (b) Rotating the $x'y'-$plane by $\theta_{cs}$ so that $(c_x, c_y, c_z)$ is located on the $\tilde{y}'-$axis.}
\end{figure}

Fig. 3 illustrates the Rx UCA center-shift. Here the origin $(0,0,0)$ and $(c_x,c_y,c_z)$ represent the coordinates of the centers of the Tx and Rx UCAs, respectively. The center of the Rx UCA is supposed to be located at the coordinate $(0, 0, D_A)$, but it is shifted to $(c_x,c_y,c_z)$ because of misalignment. For convenience, we define a vector ${\mathbf{c}} = [c_x,c_y,c_z]^T$ from the origin to the center of the Rx UCA. 
The magnitude of ${\mathbf{c}}$ is equal to $D$, which is the distance between the centers of the Tx and Rx UCAs. The direction of center-shift can be represented by two angles: the polar angle measured from the $z-$axis is denoted by $\phi_{cs}$, and the azimuthal angle of the orthogonal projection of ${\mathbf{c}}$ on the $xy-$plane measured from the $y-$axis is denoted by $\theta_{cs}$.  To represent the coordinate change caused by the center-shift using a single rotation matrix, we introduce a new coordinate system, $(\tilde{x},\tilde{y},z)$, where the $\tilde{x}-$ and $\tilde{y}-$axes are the axes obtained by rotating the ${x}-$ and ${y}-$axes by $\theta_{cs}$ about the $z-$axis. In the same way, we also define the $\tilde{x}'-$ and $\tilde{y}'-$axes from the ${x'}-$ and ${y'}-$axes. Then the center of the Rx UCA is located on the $\tilde{y}'-$axis as shown in Fig. 3(b), and its coordinates are succinctly represented as follows.
\begin{Lemma} In the new coordinate system $(\tilde{x},\tilde{y},z)$, the center of the Rx UCA is located at $(0, D\sin\phi_{cs}, D\cos\phi_{cs})$. 
	\begin{proof}
		The coordinates of the center of the new coordinate system are given by $(\tilde{c}_x,\tilde{c}_y,c_z)$ where
		$[\tilde{c}_x,\tilde{c}_y, c_z]^T=\mathbf{R}_{\theta_{cs}}^{xy}[c_x,c_y,c_z]^T=[0,\sqrt{c_x^2+c_y^2}, c_z]^T=[0,D\sin\phi_{cs},D\cos\phi_{cs}]^T$. 
		Here, the $2^{nd}$ and $3^{rd}$ equalities follow from the facts that $\cos\theta_{cs}={c_y \over \sqrt{{c_x}^2+{c_y}^2}}$, $\sin\theta_{cs}={c_x \over \sqrt{{c_x}^2+{c_y}^2}}$, and $\sqrt{c_x^2+c_y^2}=D\sin\phi_{cs}$ (Fig. 3(a)). In addition, $c_z=D\cos\phi_{cs}$. This completes the proof.
	\end{proof}
\end{Lemma}

\subsection{Modeling concurrent misalignments}

Without loss of generality, we can model the Rx UCA with all three types of misalignments, rotation, tilting and center shift, as follows. The Rx UCA is rotated and tilted at the origin and then shifted so that it is centered at $(c_x,c_y,c_z)$. From (5), the coordinate of the $n^{th}$ Rx antenna after all three misalignments, denoted as $(a_x^n, a_y^n, a_z^n)$ can be written as
\begin{equation}
\begin{split}
[a_x^n, a_y^n, a_z^n]^T&=\mathbf{c}+\mathbf{R}_{\varphi_x}^{xz}\mathbf{R}_{\varphi_y}^{yz} \mathbf{R}_{\theta_o}^{xy}[R_r\cos\theta_n, R_r\sin\theta_n, 0]^T\\
&=\mathbf{c}+\mathbf{R}_{\varphi_x}^{xz}\mathbf{R}_{\varphi_y}^{yz}\\ &\quad \quad \times [R_r\cos(\theta_n+\theta_o), R_r\sin(\theta_n+\theta_o), 0]^T.
\end{split}
\end{equation}
In the new coordinate system $(\tilde{x},\tilde{y},z)$ considered in {Lemma 1}, (6) is rewritten as follows.

\begin{Lemma}
	In the new coordinate system $(\tilde{x},\tilde{y},z)$, the coordinate of the $n^{th}$ Rx antenna, denoted as $(a_{\tilde x}^n, a_{\tilde y}^n, a_{z}^n)$, is given by 
	\begin{equation}
	\begin{split}
	[a_{\tilde x}^n, a_{\tilde y}^n, a_{z}^n]^T&=\mathbf{R}_{\theta_{cs}}^{xy}[a_{x}^n, a_{y}^n, a_{z}^n]^T\\
	&=\begin{bmatrix}
	R_1\cos({\theta_n-\alpha_1}) \\
	D\sin{\phi_{cs}}+R_2\cos({\theta_n-\alpha_2}) \\
	D\cos{\phi_{cs}}+R_3\cos({\theta_n-\alpha_3}) \\
	\end{bmatrix},
	\end{split}
	\end{equation}
	where $R_i=R_r\sqrt{b_{i1}^2+b_{i2}^2}$, ${\alpha_i}=\tan^{-1}({b_{i2} \over b_{i1}})-\theta_o$ and $b_{ij}$ is the $(i,j)^{th}$ elements of $\mathbf{R}_{\theta_{cs}}^{xy}\mathbf{R}_{\varphi_x}^{xz}\mathbf{R}_{\varphi_y}^{yz}$. 
\end{Lemma}
The result in (7) can be obtained by directly calculating $\mathbf{R}_{\theta_{cs}}^{xy}[a_{x}^n, a_{y}^n, a_{z}^n]^T$.
The coordinate of the $m^{th}$ Tx antenna in the new coordinate system is given by $(R_t\cos\theta_m',R_t\sin\theta_m',0)$ where $\theta_m'=\theta_m+\theta_{cs}$. Now the distance between the $m^{th}$ Tx antenna and the $n^{th}$ Rx antenna, $d(n,m)$, can be represented as follows.

\begin{Lemma} The distance $d(n,m)$ is written as 
	\begin{equation}
	d(n,m)= D\big\{1+f(D,R_t,R_r,\theta_m,{\theta_n},{\theta_o},{\theta_{cs}},\varphi_x,\varphi_y,\phi_{cs})\big\}^{1 \over 2},
	\end{equation} 
	where $f(D,R_t,R_r,\theta_m,\theta_n,\theta_o,\theta_{cs},\varphi_x,\varphi_y,\phi_{cs})$ 
	is presented in Appendix A. \end{Lemma}
\begin{proof}
	See proof in Appendix A.
\end{proof}
Because $D \gg R_t$ and $R_r$, $d(n,m)$ can be approximated as 
\begin {equation}
\begin{split}
	d(n,m) =&D\Big\{ 1+{1\over {2}}f(D,R_t,R_r,\theta_m,{\theta_n},{\theta_o},{\theta_{cs}},\varphi_x,\varphi_y,\phi_{cs})\Big\}\\
	=&d_A(n,m)-\tau_{t}(m)+\tau_{r}(n),
	\end {split}
	\end {equation} 
	where 
	\begin{equation}
	d_A(n,m)={D-{R_tR_r \over D}\cos({\theta_n}-{\theta_m}+{\theta_o})},
	\end{equation}
	\begin{equation}
	\tau_{t}(m)=R_t\sin{(\theta_m+\theta_{cs})}\sin{\phi_{cs}},
	\end{equation}
	and 
	\begin{equation}
	\begin{split}
	\tau_{r}(n)&={R_r^2\over 4D}\{\cos2({\theta_n-\alpha_1})+\cos2({\theta_n-\alpha_2})
	+\cos2({\theta_n-\alpha_3})\}\\
	&\quad +R_2\cos({\theta_n-\alpha_2})\sin\phi_{cs}+R_3\cos({\theta_n-\alpha_3})\cos\phi_{cs}.
	\end{split}
	\end{equation}
	Note that in (9), $d(n,m)$ is decomposed into three components defined in (10)\textendash(12). Comparing (10) with (3), we can see that $d_A(n, m)$ can be thought of as the distance between the Tx and Rx antennas of an aligned UCA system that only has rotation, such as the one shown in Fig. 1. In (11), $\tau_t(m)$  represents the displacement caused by the center-shift only, while the displacement $\tau_r(n)$ in (12) is caused by all three types of misalignments. The subscripts $t$ and $r$ of $\tau_t(m)$ and $\tau_r(n)$ indicate that $m$ and $n$ are the indices of Tx and Rx antennas, respectively. In what follows, we shall see that the expression for $d(m, n)$ in (9) leads to an efficient representation for the misaligned channel matrix.
	
	Consider the misaligned channel matrix $\mathbf{H} \in \mathbb{C}^{N_s \times N_s}$ whose $(m, n)^{th}$ entry is given by
	\begin {equation}
	\begin {split}
	{h}(n,m) &= e^{-j{2\pi \over \lambda} d(n,m)}\\
	&=h_A(n,m)\cdot T_{t}^*(m)\cdot T_{r}(n)
	\end {split}
	\end {equation}
	which is the normalized free-space channel response, where  $\lambda$ is the wavelength of the carrier [2]\textendash[4] and the 2nd equality follows from (9); $h_A(n,m)=e^{-j{2\pi \over \lambda}D}\cdot e^{+j{\beta}\cos({\theta_n}-{\theta_m}+{\theta_o})}$, $T_{t}(m)=e^{-j{2\pi \over \lambda}\tau_{t}(m)}$, and $T_{r}(n)=e^{-j{2\pi \over \lambda}\tau_{r}(n)}$. Here, $\beta\triangleq{2\pi R_tR_r \over \lambda D}$. The parameter $\beta$ is referred to as the RPDR.
	In matrix-form, (13) is rewritten as
	\begin {equation}
	\begin {split}
	\mathbf{H}&= \mathbf{T}_{r}\mathbf{H}_A\mathbf{T}_{t}^H,
	\end {split}
	\end {equation}
	where $\mathbf{T}_{t}=diag[T_t(1), \cdots ,T_t(N_s)]$ and $\mathbf{T}_{r}=diag[T_r(1), \cdots ,T_r(N_s)]$; the $(n, m)^{th}$ entry of the matrix $\mathbf{H}_A$ is given by $h_A(n, m)$. 
	Because the matrix $\mathbf{H}_A$ is a circulant matrix, (14) can be further decomposed into 
	\begin {equation}
	\begin {split}
	\mathbf{H} = {\mathbf{T}_{r}}\mathbf{Q}{\mathbf{\Delta}_A}{\mathbf{Q}^H}\mathbf{T}_{t}^H,
	\end {split}
	\end {equation}
	where $\mathbf{Q}$ is the DFT matrix, $\mathbf{H}_A=\mathbf{Q}{\mathbf{\Delta}_A}{\mathbf{Q}^H}$, and $\mathbf{\Delta}_A \in \mathbb{C}^{N_s \times N_s}$ is a diagonal matrix. 
	Let $\mathbf{\Delta}_A=\mathbf{S}|\mathbf{\Delta}_A|$, where $\mathbf{S} \in \mathbb{C}^{N_s \times N_s}$ is a diagonal matrix of complex numbers having unit magnitude. Then, the decomposition of the channel matrix $\mathbf{H}$ in (15) leads to the following property.  
	\begin{Property} 
		Let $\mathbf{H}= \mathbf{U} \mathbf{\Sigma} \mathbf{V}^H$ represent the SVD of $\mathbf{H}$, where $\mathbf{U}\in \mathbb{C}^{N_s \times N_s}$ and $\mathbf{V}\in \mathbb{C}^{N_s \times N_s}$ are unitary matrices, and $\mathbf{\Sigma}\in \mathbb{R}^{N_s \times N_s}$ is a diagonal matrix with singular values on its diagonal. 
		Then, the SVD of $\mathbf{H}$ can be written as
		\begin{equation}
		\begin{split}
		\mathbf{U} &= {\mathbf{T}_{r}}\mathbf{QS},\\
		\mathbf{\Sigma} &= |\mathbf{\Delta}_A|,\\
		\mathbf{V} &=\mathbf{T}_{t} {\mathbf{Q}}. 
		\end{split}
		\end{equation}
		\begin{proof} 
			Because $\mathbf{T}_t$, $\mathbf{T}_r$, and $\mathbf{S} $ are diagonal matrices with unit gain $(|\mathbf{T}_t| = |\mathbf{T}_r| = |\mathbf{S}| =\mathbf{I}_{N_s})$, $\mathbf{T}_r\mathbf{Q}\mathbf{S}$ and $\mathbf{T}_t\mathbf{Q}$ are unitary matrices that can serve as the left and right singular matrices of the SVD, and the diagonal entries of $\mathbf{\Sigma}$ are the singular values. Here the singular values are not sorted in order.
		\end{proof}
	\end{Property} 
	
	Because of Property 1, the singular values of $\mathbf{H}$ are identical to those of $\mathbf{H}_A$; thus, the capacities of $\mathbf{H}$ and $\mathbf{H}_A$ are the same. Consequently, the singular values and the capacity of the misaligned channel $\mathbf{H}$  are \textit{independent of} the tilting and center-shift angles $\{{\theta_{cs}, \phi_{cs}, \varphi_x, \varphi_y}\}$. Next we analyze the singular values and obtain the optimal radii of the UCAs maximizing the capacity.

\section{Capacity and Optimal Design}

\subsection{Singular Value Analysis}

Because a singular value of $\mathbf{H}$, denoted as $\sigma_{k}$ for $k \in \{1,2,\cdots, N_s\}$, is equal to that of $\mathbf{H}_A=\mathbf{Q}\mathbf{\Delta}_A\mathbf{Q}^H$, $\sigma_{k}$ can be represented as 
\begin{equation}
\begin{split}
\sigma_{k}&=\Big|\mathbf{Q}(:,k)^H\mathbf{H}_A\mathbf{Q}(:,k)\Big|\\
&=\Big|\sum_{i=0}^{N_s-1}e^{-j\{{2\pi \over N_s}i(k-1)-\beta\cos{({2\pi \over N_s}i+\theta_o)}\}}\Big|,
\end{split}
\end{equation}
and $\sum_{k=1}^{N_s}\sigma_{k}^2=\mathrm{tr}(\mathbf{H}_A\mathbf{H}_A^H)=N_s^2$.  The singular value is a function of $N_s$, $\theta_o$ and $\beta$. In our analysis, we assume that $N_s$ is given and, whenever necessary, the singular values are denoted as  $\sigma_{k}(\beta, \theta_o)$. The singular values $\sigma_{k}(\beta, \theta_o)$ exhibit the following characteristics.
\begin{Property} Suppose that $N_s$ is an even number. 
	\begin{enumerate}[label=(\alph*)] 
		\item  $\sigma_{k}(\beta, \theta_o)=\sigma_{{N_s+2-k}}(\beta, \theta_o)$ for $\{ k |2 \leq k \leq N_s\}$. 
		\item $\sigma_{1}(\beta, \theta_o)\geq \sigma_{k}(\beta, \theta_o)$ for $\{ k |2 \leq k \leq N_s\}$, if $0 \leq \beta \leq {\pi N_s \over 4\sum_{i=0}^{N_s-1} |\cos{({2\pi  \over N_s}i+\theta_o)}|}$.
		\item $\lim_{\beta\to0} \sigma_{1}(\beta, \theta_o)=N_s$ and $\lim_{\beta\to0} \sigma_{k}(\beta, \theta_o)=0$ for $\{ k |2 \leq k \leq N_s\}$.
		\item $\sigma_{k}(\beta, \theta_o)=\sigma_{k}(\beta, -\theta_o)$.
		\item $\sigma_{k}(\beta, \theta_o)=0$  if $k=N_s/2+1$ and  $\theta_o=\pm\pi/N_s$.
	\end{enumerate}
\end{Property}

The proofs for Properties 2\textit{(a)\textendash(e)} are presented in Appendix B. Property 2(\textit{a}) is an extension of the property in [10] showing that $\sigma_k(\beta) = \sigma_{N_s+2-k}(\beta)$ when $\theta_o = 0$.  Because of Property 2(\textit{a}), there are $N_s/2 + 1$ distinct singular values for given $N_s$ and $\beta$. Properties 2(\textit{b}) and 2(\textit{c}) indicate that $\sigma_1$ becomes dominant as the RPDR, $\beta$, decreases. When the radii product $R_t R_r$ is fixed, $\beta$ decreases as the communication range $D$ increases. Therefore, if a UCA system with small antennas is deployed for long-distance communication, then only one data stream can be transmitted without multiplexing. Property 2(\textit{d}) indicates that the effects of clockwise and counterclockwise rotations on the singular values are the same. Finally, Property 2(\textit{e}) shows that $\sigma_{N_s/2+1}(\beta, \theta_o)$ for the $(N_s/2+1)^{th}$ eigen-mode becomes zero when $|\theta_o|$ hits its maximum value $\pi/N_s$.
\begin{figure}
	\centering
	\subfloat[]{\includegraphics[width=0.25\textwidth, height=0.16\textwidth]{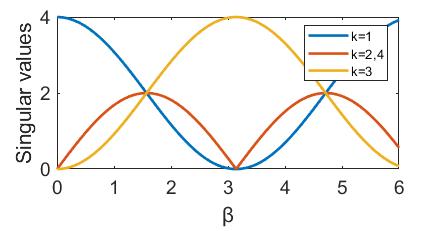}}
	\subfloat[]{\includegraphics[width=0.25\textwidth, height=0.16\textwidth]{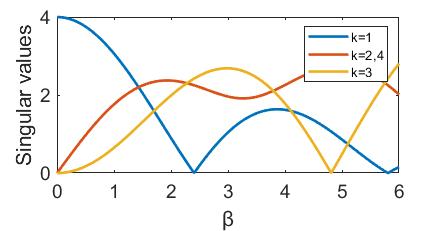}}\\
	\subfloat[]{\includegraphics[width=0.25\textwidth, height=0.16\textwidth]{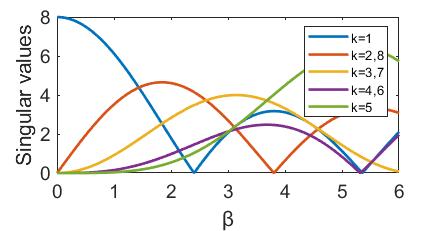}}
	\subfloat[]{\includegraphics[width=0.25\textwidth, height=0.16\textwidth]{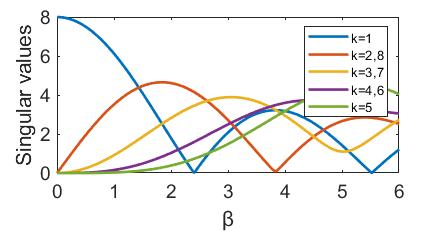}}
	\caption{Singular values against $\beta$. (a) $N_s=4$ and $\theta_o=0$. (b) $N_s=4$ and $\theta_o=\pi/(2{N_s})$. (c) $N_s=8$ and $\theta_o=0$. (d) $N_s=8$ and $\theta_o=\pi/(2{N_s})$.}
\end{figure}

\begin{figure}
	\centering
	\subfloat[]{\includegraphics[width=0.25\textwidth, height=0.16\textwidth]{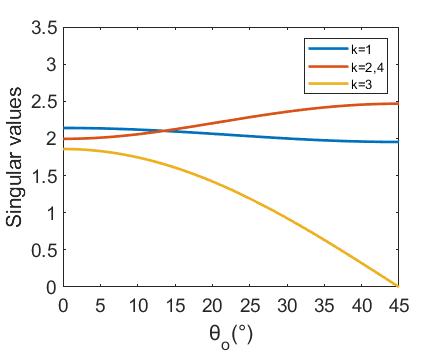}}
	\subfloat[]{\includegraphics[width=0.25\textwidth, height=0.16\textwidth]{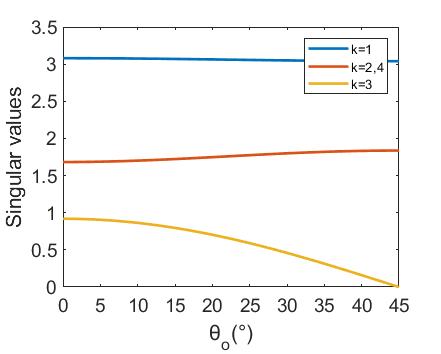}}\\
	\subfloat[]{\includegraphics[width=0.25\textwidth, height=0.16\textwidth]{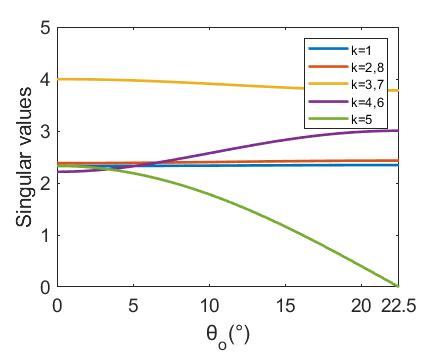}}
	\subfloat[]{\includegraphics[width=0.25\textwidth, height=0.16\textwidth]{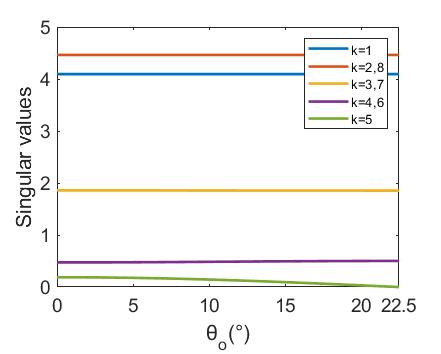}}
	\caption{Singular values against $\theta_o$. (a) $N_s=4$ and $\beta=1.5$. (b) $N_s=4$ and $\beta=1$. (c) $N_s=8$ and $\beta=3.1$. (d) $N_s=8$ and $\beta=1.5$.}
\end{figure}

Figs. 4 and 5 show the singular value curves against $\beta$ and $\theta_o$, respectively. These curves confirm Properties 2(\textit{a})\textendash(\textit{e}). The singular values fluctuate as $\beta$ varies, but they tend to vary slowly for $\theta_o$. In the following subsection, we shall see that the system capacity also fluctuates with $\beta$, and we will find the optimal value of $\beta$ that maximizes the capacity.

\subsection{Optimal Radii of UCAs}
In [13], a one-dimensional search process for obtaining the optimal radii of an aligned UCA system, maximizing the spectral efficiency when $N_s$, $\lambda$, and $D$ are given, was  proposed. This method searches for the optimal value of the product of $R_t$ and $R_r$, under the assumption that equal power allocation is adopted. In this subsection, we present an alternative approach to determining the optimal radii. Throughout this subsection, it is assumed that $\theta_o$ is fixed. The proposed method assumes the water-filling power allocation and maximizes the capacity; however, it is simpler to implement than the method introduced in [13].

\begin{figure}[t]
	\centering
	\centerline{\includegraphics[width=0.43\textwidth, height=0.26\textwidth]{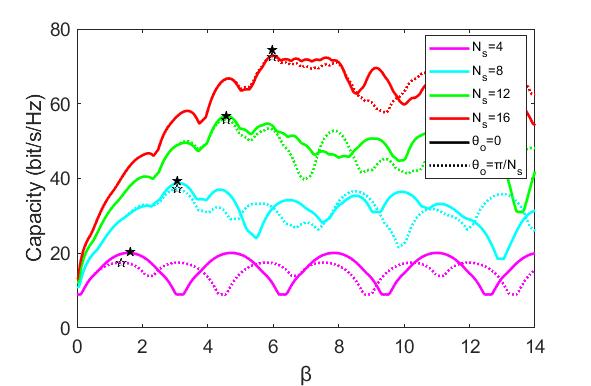}}
	\caption{Capacity curves ${\mathcal{C}}_{N_s}$ against $\beta$ when $N_s \in \{4,8, 12, 16\}$ and $P_T/N_o=15$dB. The maximum values ${\mathcal{C}}_{N_s}$ are marked by  $\bigstar$ for $\theta_o=0$ and $\largestar$ for $\theta_o=\pi/Ns$.}
\end{figure}

The capacity of a UCA system, denoted as $\mathcal{C}_{N_s}$, with channel $\mathbf{H}$ corrupted by additive white Gaussian noise (AWGN), is given by 
\begin{equation}
{\mathcal{C}}_{N_s}=\sum_{k=1}^{N_s}\log_2(1+{ p_{k}\sigma_{k}^2 \over {N_o}})
\end{equation}
where $\{p_{k}\}$ denotes the water-filling power allocations [20]. The proposed scheme is based on the observation that both $\sigma_{k}$ and $p_{k}$ in (18) are functions of $N_s$ and $\beta$. This observation holds true because the optimal power allocation $\{p_{k}\}$ are functions of $\{\sigma_{k}\}$, which in turn are functions of $N_s$ and $\beta$, as shown in (17). Based on this observation, we can search the an optimal RPDR value $\beta^o$ that maximizes the capacity in (18) when $N_s$ and the signal-to-noise ratio (SNR), $P_T/ N_o$ where $P_T = \sum_{k=1}^{N_s} p_{k}$, are given. The one-dimensional search for finding $\beta^o$ is straightforward. The search range starts from zero and is reasonably narrow, because $D \gg R_t$ and $R_r$. In the proposed design, we pre-determine the optimal RPDR $\beta^o$ values for all $\{N_s, {P_T\over N_o}\}$ values of interest during the initial stage and use them to obtain the optimal radii product $R_tR_r$ when $\lambda$ and $D$ are given. This two-step process is simpler to implement than this proposed in [13], which  directly searches for the optimal radii product $R_tR_r$. 
\begin{table}[tp]\label{aggiungi }\centering
	\caption{Optimal RPDR values ($\beta^o$) for $N_s\in\{4,8,12,16\}$ {and $\theta_o=0$}.}
	\begin{tabular}{ c||c|c|c|c }
		\hline                   
		\backslashbox{SNR (dB)}{$N_s$} &4&8 &12&16 \\
		\hline
		\hline
		$5$& 1.57&3.10&4.53&5.98\\
		
		$10$& 1.51&3.08&4.56&5.97\\
		
		$15$& 1.54&3.09&4.57&5.98\\
		
		$20$& 1.54&3.08&4.55&5.98\\
		\hline
	\end{tabular}\\
\end{table}

Fig. 6 shows the capacity curves ${\mathcal{C}}_{N_s}$ against $\beta$, for $0 < \beta \leq 14$, $N_s \in \{4, 8, 12, 16\}$, $\theta_o \in \{0, \pi/N_s\}$, and $P_T/N_o = 15$ dB. When $\beta$ approaches zero, because of Property 2(\textit{c}), the capacity of a UCA system in (18) becomes $\log_2(1+{ P_TN_s^2 \over {N_o}})\approx\log_2($SNR$)+2\log_2{N_s}$. In general, the optimal RPDR, $\beta^o$, is not unique and we choose the smallest $\beta^o$ from the set of the optimal RPDR values, to minimize the antenna size. Note that the optimal RPDR values maximizing the capacity for $\theta_o = 0$ and $\pi/N_s$ are almost identical; they are robust to $\theta_o$, and we list only the optimal RPDR values for $\theta_o =0$ in Table I, which shows $\beta^o$ for $P_T/N_o \in$ {$\{5$ dB, $10$ dB, $15$ dB, $20$ dB$\}$}. In Table I, the optimal RPDR values  for a fixed $N_s$ are almost identical irrespective of the SNR. This happens because the singular values are independent of the SNR, and equal power allocation is almost optimal under a high-SNR regime. Furthermore, as seen in Table I, $\beta^o$ tends to increase linearly with $N_s$. Therefore, when designing a UCA system for a given $D$, it is necessary to increase the product $R_tR_r$ in proportion to $N_s$.

To illustrate UCA systems designed by the proposed method, we design UCA systems with the following parameters: $\lambda = 0.004$ meters (75 GHz), $D = 100$ meters, $R_t = R_r$, SNR = $15$ dB, and $N_s \in \{4, 8, 12, 16\}$. The results are shown in Table II. The optimal radius increases with $N_s$ and the capacity gain achieved by increasing both $N_s$ and the radius can be significant.\footnote{In fact, the capacity tends to increase linearly with $N_s$. This can be seen from the following capacity upper bound, which is valid under a high-SNR regime [20]: ${\mathcal{C}}_{N_s} \leq N_s \log{({P_T \over N_s^2N_o}\sum_{k=1}^{N_s}\sigma_{k}^2)}= N_s \log{P_T \over N_o}$, where the equality holds due to (17). 
}
\begin{table}[tp]\centering
	\caption{Optimal radius ($R_t = R_r$) and the corresponding capacity when $N_s\in\{4,8,12,16\}$, $\lambda = 0.004$ meters (75GHz),  $D = 100$ meters, and SNR = $15\mathrm{dB}.$}
	\begin{tabular}{ c||c|c|c|c }
		\hline                      
		{$N_s$} &4&8 &12&16 \\
		\hline
		\hline
		{$R_t(R_r)$} (meters) & 0.31 & 0.44 & 0.54 &  0.62\\\hline
		{${\mathcal{C}}_{N_s}$} (bit/s/Hz)& 20.11 &38.79&56.79&72.88\\
		\hline
	\end{tabular}\\
\end{table}
\subsection{Optimal Transceiver of the Optimal UCA System}

Table \rom{3} shows the condition numbers ($\max \sigma_k/\min \sigma_k$) when the RPDR value is optimal and $\theta_o = 0$. 
When $N_s > 4$, although the singular values are not identical when the RPDR value satisfies the optimal value, the deviation among singular values is reasonably small if the number of antennas is small, e.g., $N_s = 8$, indicating that the column vectors of the channel $\mathbf{H}$  are nearly orthogonal to each other. Therfore, the ZF receiver can  achieve almost  maximum channel capacity without precoding at the optimal criteria. When the number of antennas is large, the channel capacity is achieved by the ZF receiver in conjunction with the successive interference cancellation (SIC) of data streams.

\begin{table}[tp]\centering
	\caption{Condition number of $\sigma_k$ when $N_s\in\{4,8,12,16\}$ and $\theta_o=0$}
	\begin{tabular}{ c|c||c|c|c|c }
		\hline                      
		\multicolumn{2}{c||}{$N_s$} &4&8 &12&16 \\
		\hline
		\hline
		Condition number   &$\beta=\beta^o$& 1 & 1.84 & 2.42 & 3.51\\ 
		($\max \sigma_k \backslash \min \sigma_k$)&$\beta=0.5\beta^o$& 6.36 & 22.63 & 104.53 & 469.97\\\hline
	\end{tabular}\\
\end{table}

\section{Precoder Design for Non-optimal UCA System}

Although the optimal design criterion allows us to achieve the maximum capacity, UCA systems that consider longer transmission distances with a limited array size are also preferred from practical implementation perspectives. If the communication distance is longer than the optimal value (or the radii are smaller than the optimal value), the deviation between singular values is considerable as shown in the case of $\beta = 0.5\beta^o$  in Table \rom{3}. Hence precoding is necessary for achieving high data rates. 
\subsection{Codebook Based Precoder}
Referring to Property 1, the SVD-based capacity achieving the optimal  precoding scheme of a UCA system is $\mathbf{V}=\mathbf{T}_t\mathbf{Q}$ with the water-filling power allocation where $\mathbf{T}_t$ is determined by the center-shift angles $\{\theta_{cs}, \phi_{cs}\}$. To implement this optimal precoder,  the information of $\{\theta_{cs}, \phi_{cs}\}$ and  $\theta_o$ is required for precoding and power allocation, respectively; otherwise, full channel-state information is necessary. Because the center-shift angles can be viewed as the directions of arrival (DoAs) when the center of the Tx UCA is regarded as a source point, existing estimation techniques, such as the multiple-signal classification (MUSIC) algorithm can be used to obtain those angles. However, tilting of the Rx UCA induces phase offsets of the center-shift angles (or DoAs), which degrade estimation accuracy (we demonstrate this through computer simulations in the next section).

An alternative approach is to construct a codebook with angular quantizations. We quantize the center-shift angles so that their sine values are uniformly distributed within ranges determined by the angular ranges of the center-shift angles. This is because the phase term of $\mathbf{T}_t$ in (11) is proportional to  the sines of $\theta_{cs}$ and $\phi_{cs}$. Because $\phi_{cs}$ represents the degree of shift, an angular range of $\phi_{cs}$ can be narrow in wireless backhaul scenarios in which deviations might be small. For example, when $D=100$ meters and the Rx UCA moves $10$ meters from the $z-$ axis and then $\phi_{cs}=0.1$ rad $(5.73^\circ$),  the angular range of $\phi_{cs}$  can be set as $[-0.175$ rad $(-10^\circ),  0.175$ rad $(10^\circ)]$. However, $\theta_{cs}$ represents the direction of the center-shift, which ranges from $-\pi$ to $\pi$. Thus, its angular range should be $[-\pi/2, \pi/2]$ because we quantize $\sin{\theta_{cs}}$, and $\sin{(\pi-\theta_{cs})}=\sin{\theta_{cs}}$. 

Let the center-shift angles $\{\theta_{cs}, \phi_{cs}\}$ be quantized to $2^{L_1}$ and $2^{L_2}$ elements, respectively. Then we have $2^{L}$ sets of angles, where $L=L_1 + L_2$. We denote the set of angles as $\overline{\mathcal{A}}_{(l)}=\{{\overline{\theta}_{cs}}_{(l_1)}, {\overline{\phi}_{cs}}_{(l_2)}\}$, where $l=1,\cdots, 2^{L}$ and ${\overline{\theta}_{cs}}_{(l_1)}$ (${\overline{\phi}_{cs}}_{(l_2)}$) is the $l_1 ^{th}$ ($l_2 ^{th}$) quantized angle of $\theta_{cs}$ ($\phi_{cs}$) among $2^{L_1}$ ($2^{L_2}$) elements. The optimal set of angles is determined by a selection algorithm operating at Rx; and the Rx transfers  $L$ bits index of the optimal set to the Tx. Let ${\overline{\mathbf{T}}_{t}}_{(l)}$ denote ${\mathbf{T}}_{t}$ obtained by the quantized angles from the set $\overline{\mathcal{A}}_{(l)}$. We consider the maximum achievable rate as a performance metric, and then the receiver can search for the codebook index that solves
\begin{equation}
l_{opt}=\underset{0\leq l\leq 2^L}\argmax \quad \log_2\det(\mathbf{I}_{N_s}+\mathbf{H}\overline{\mathbf{F}}_{(l)}\mathbf{P}{\overline{\mathbf{F}}_{(l)}^H}\mathbf{H}^H).
\end{equation}
Here, $\mathbf{P}$ is a diagonal matrix whose $k^{th}$ diagonal entry is given by $p_k/N_o$  and $\overline{\mathbf{F}}_{(l)}\triangleq {\overline{\mathbf{T}}_{t}}_{(l)}\mathbf{Q}$. The water-filling power allocations $\{p_k\}$ can be approximated by the following  power allocation policy.


{\subsection{Approximate Power Allocation}
 Because of the robustness of singular values against $\theta_o$, the power allocations of the misaligned UCA system can be substituted with the power allocations designed for the aligned UCA system where $\theta_o=0$. This roughly performed power allocation only requires the information of $D$ at given $R_t$, $R_r$, $\lambda$, and $N_s$, without cumbersome estimation of $\theta_o$, but only minor performance degradation appears, which will be demonstrated via computer simulations in Section \rom{5}.}

\begin{figure}
	\centering
	\subfloat[]{\includegraphics[width=0.21\textwidth, height=0.07\textwidth]{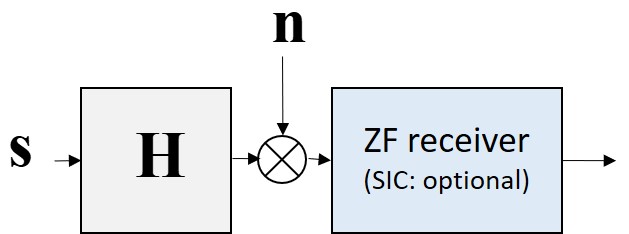}}\\
	\subfloat[]{\includegraphics[width=0.38\textwidth, height=0.098\textwidth]{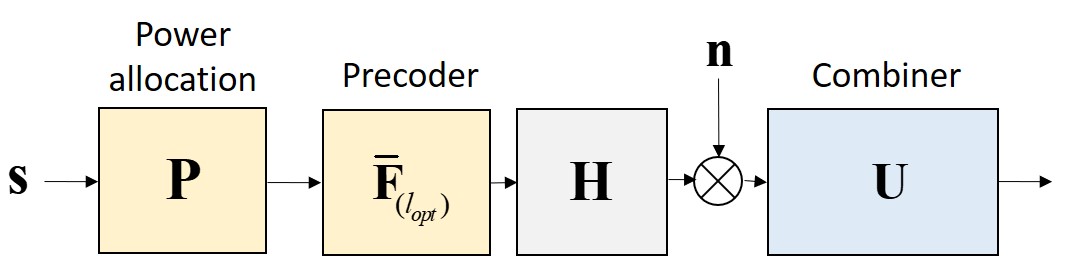}}
	\caption{ Proposed transceiver architectures for misaligned UCA systems. (a) At the optimal criterion. (b) When the optimal criterion is not satisfied.}
\end{figure}
\vspace{5mm}
The proposed transceiver architectures for the optimal and  non-optimal UCA systems to achieve the channel capacity without estimating or compensating misalignment angles are illustrated in Fig. 7. The optimally designed UCA systems and the non-optimal UCA systems are complementary to each others. The optimal UCA systems achieve the maximum capacity by the simple ZF receiver  at the expense of large UCA radii, whereas the non-optimal UCA systems can fit in a small space but require precoding and suffer from capacity reduction. 
Two natural questions from an implementation point of view are how far the ZF receiver can support the capacity and when precoding is worthwhile at the fixed UCA radii.
To answer these questions, we investigate the performance of the two transceivers through computer simulations.

\section{Simulation Results}

The performance of the proposed transceivers of a misaligned UCA system is examined through computer simulations with the following parameters: the carrier frequency is $75$ GHz ($\lambda=2mm$), SNR=15 dB, and the radii of the Tx/Rx UCAs are assumed to be equal. i.e., $R_t=R_r$, and fixed to the optimal value at the communication distance of $100$ meters. All results are obtained by averaging over 100 channel realizations, in which  misalignment angles are randomly generated to have a uniform distribution with a zero mean. We assume that the misalignment angles   $\{\varphi_{x},\varphi_{y}, \phi_{cs}, \theta_o\}$ are within the interval [$-{10\over 180} \pi$, ${10\over 180} \pi$], except for $\theta_{cs}$ which is in the interval [$-\pi$, $\pi$]. 

This section consists of two parts. In the first part, we examine the performance of  proposed codebook-based precoder and approximated water-filling power allocation policy. Then, in the second part, we compare the performance of two transceiver systems in Fig. 7 to provide guidelines for transceiver implementation.

{\subsection{Performance Evaluation of Proposed Precoding scheme}
	The performance of the proposed codebook-based precoder is compared with that of its benchmarks: the optimal precoder, the precoder obtained by the MUSIC algorithm, and the identity precoder. The metric for the performance comparison is the maximum achievable rate. For the MUSIC estimation, we assumed that an antenna on the Tx UCA acts as a source point so that the source signal transmitted from the $m^{th}$ Tx antenna is conveyed through the channel corresponding to the $m^{th}$ column vector of $\mathbf{H}$. Then, the Rx computes the MUSIC algorithm using the array response vector of the UCA [21]. Since the  difference between the $m^{th}$ column vector of $\mathbf{H}$ and the array response vector of the UCA mainly comes from tilting, this setting reflects the effects of tilting on the MUSIC estimation.  
	\begin{figure}
		\centering
		\subfloat[]{\includegraphics[width=0.26\textwidth, height=0.22\textwidth]{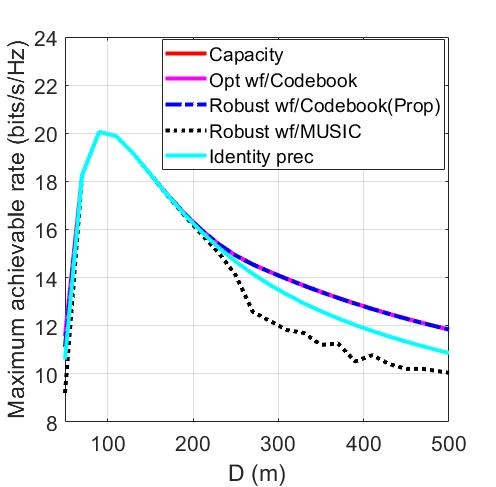}}
		\subfloat[]{\includegraphics[width=0.26\textwidth,height=0.22\textwidth]{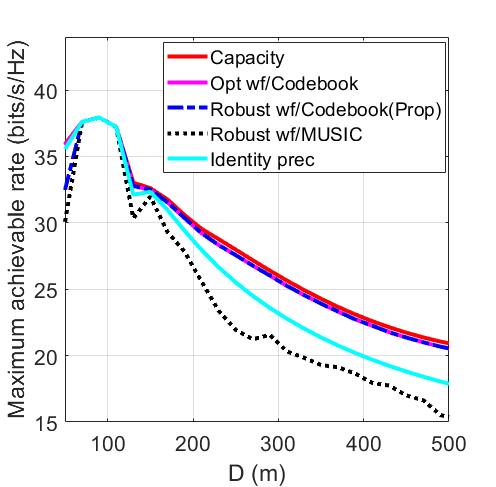}}\\
		\subfloat[]{\includegraphics[width=0.26\textwidth, height=0.22\textwidth]{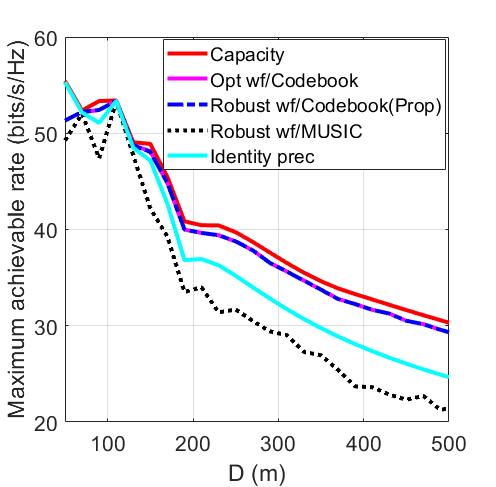}}
		\subfloat[]{\includegraphics[width=0.26\textwidth, height=0.22\textwidth]{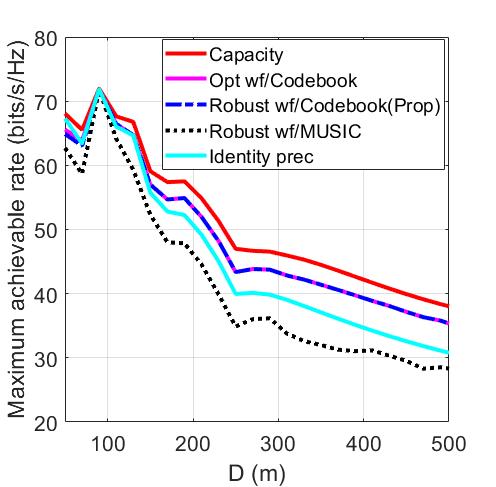}}
		\caption{ Maximum achievable rate curves of the proposed precoder and its benchmarks against $D$. The number of bits of codebook $L=8$ bits ($L_1=5$ bits and $L_2=3$ bits). (a) $N_s=4$. (b) $N_s=8$. (c) $N_s=12$. (d) $N_s=16$.}
	\end{figure}

	Fig. 8 shows the maximum achievable rate curves of the proposed precoder and its benchmarks  against $D$  when $N_s=4,8,12$ and $16$, and the codebook bits of the proposed precoder are given as $(L_1,L_2)=(5,3)$ bits. The results show that the maximum achievable rates of the proposed precoder coincide with the channel capacity when $N_s= 4$ and $8$, but when $N_s =12$ and $16$, there is a small gap between them. Because the performance of the proposed approximated water-filling power allocation is the same as that of the optimal power allocation, the approximation in computing power allocations has little effect on the performance degradation. The gap mainly comes from the quantizations in codebook design.  The gap tends to increase with the number of antennas because the greater the deviation between singular values the more sensitive the precoder performance is. Noticeable precoding gains over the identity precoder appear when the communication distance exceeds about $200$ meters. When $D=300$ meters, the proposed precoder provides about 4 bits/sec/Hz gains over the identity precoder  when $N_s\geq 12$. When $D=500$ meters, precoding gains of more than $9 \% $ appear with all numbers of antennas. 	Interestingly, the maximum achievable rate of the precoder obtained by the MUSIC algorithm is lower than that of the identity precoder. This is because the phase offsets of the center-shift angles (or DoAs) induced by  tilting of the Rx UCA result in significant performance degradation, as we discussed at the beginning of Section \rom{4}.

	\begin{figure}
		\centering
		\subfloat[]{\includegraphics[width=0.25\textwidth,height=0.25\textwidth]{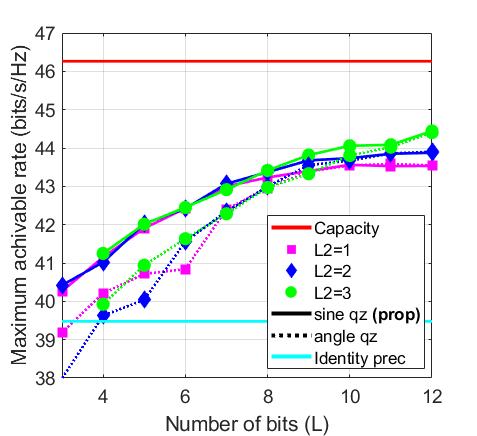}}
		\subfloat[]{\includegraphics[width=0.25\textwidth, height=0.25\textwidth]{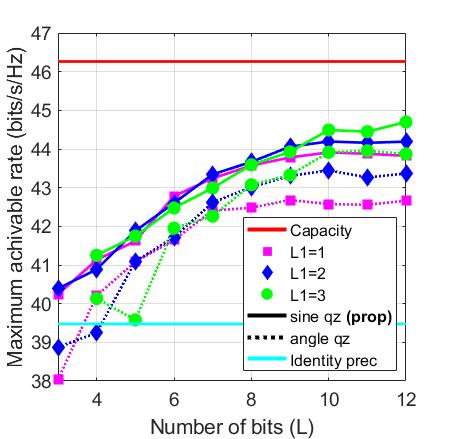}}
		
		\caption{ Maximum achievable rate curves of the proposed precoder against the number of codebook bits $L$ when $N_s=16$ and $D=300$ meters. (a) $L_2$ is fixed. (b) $L_1$ is fixed. }
	\end{figure}

	Fig. 9 shows the impacts of bit allocations and the quantization method in codebook design when $N_s=16$ and $D=300$ meters. We compare the codebook designed by the proposed quantization method in which the angles are quantized to have uniformly distributed sine values, with the codebook designed by linear quantization of the angles, and show that the proposed quantization method outperforms the linear quantization method. We also compare two types of codebook bit-allocation methods: Fig. 9(a) shows the performance of codebooks in which $L_2$ is fixed to $1,2$, and $3$ bits, respectively, and $L_1$ increases so that more bits are allocated to $\theta_{cs}$, and Fig. 9(b) shows that of the opposite case.  Comparing the codebook bit-allocations, the performance of the codebook depends on the total number of codebook bits, not the bit-allocation method. This is because, even though the angular range of $\theta_{cs}$ is much larger than that of $\phi_{cs}$,  $\theta_m$ is dominant to determine  the term $\sin(\theta_m+\theta_{cs})$ in (11); thus, $\theta_{cs}$ has a small effect on the performance compared to its degree of deviation. 	The maximum achievable rates of the proposed precoder increase linearly with $L$ when $L \leq 8$, but saturate beyond that.

	\begin{figure}
		\centering
			\subfloat[]{\includegraphics[width=0.25\textwidth, height=0.22\textwidth]{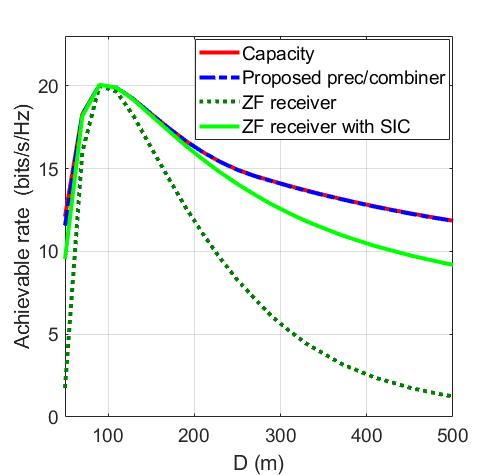}}
		\subfloat[]{\includegraphics[width=0.25\textwidth,height=0.22\textwidth]{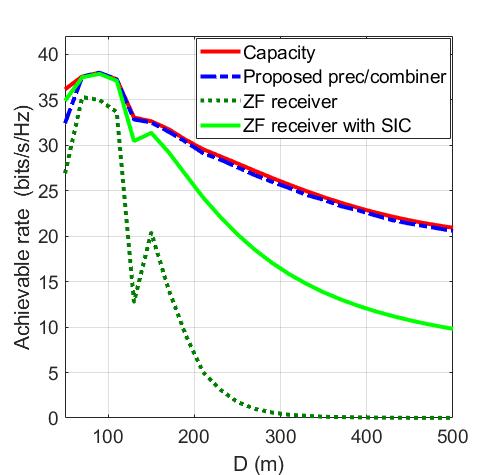}}\\
		\subfloat[]{\includegraphics[width=0.25\textwidth, height=0.22\textwidth]{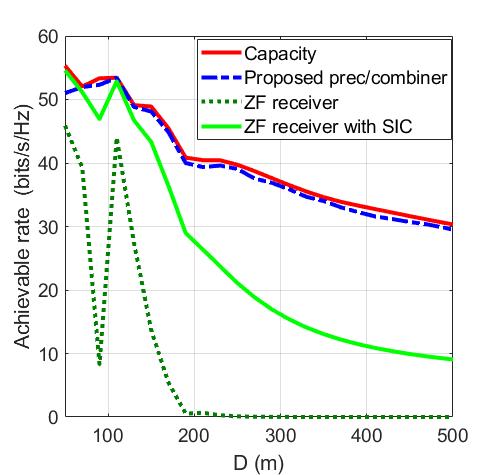}}
		\subfloat[]{\includegraphics[width=0.25\textwidth, height=0.22\textwidth]{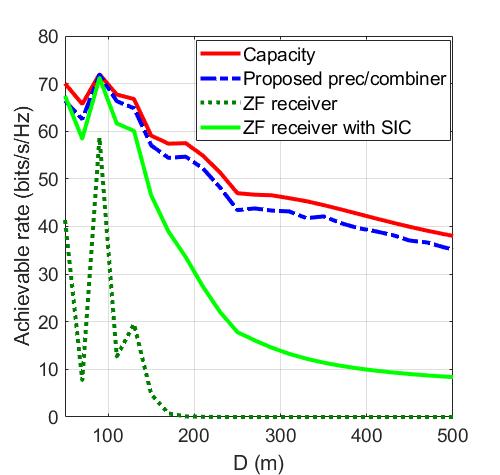}}
		\caption{ Achievable rate curves of the proposed transceiver architectures of UCA systems against $D$.  The number of bits of codebook $L=8$ bits ($L_1=5$ bits and $L_2=3$ bits). (a) $N_s=4$.  (b) $N_s=8$. (c) $N_s=12$. (d) $N_s=16$.}
	\end{figure}

	\subsection{Comparison of Proposed Transceiver Architectures}

	Fig. 10 shows the achievable rate curves of the two proposed  transceiver architectures illustrated in Fig. 7 against $D$ when $N_s\in\{4,8,12,16\}$. 	As expected, precoding gains are minor at the optimal distance ($100$ meters), and the ZF receiver achieves the channel capacity when $N_s=4$. When $N_s=8,12$, and $16$, the ZF receiver achieves the channel capacity with the help of the SIC. The ZF receiver tolerates small deviations of $D$ from its  optimal value, but the performance gap between the ZF receiver and the proposed precoder and combiner scheme is considerable when the communication distance exceeds $150$ meters.  Therefore, we propose the use of the ZF receiver when the deviation of $D$ is within a few meters; otherwise, the proposed precoder/combiner system is recommended.

	\section{Conclusion}
	
	We proposed an optimal design method and transceiver architectures for misaligned UCA systems, which can be implemented without the knowledge of misalignment angles. We derived a channel model of the misaligned UCA system, and from the derived channel model, it was  shown that the singular values of the misaligned UCA system varies with an RPDR value but is robust to other misalignments. Then, we proposed an optimal design method of UCA systems that performs a one-dimensional search of RPDR to maximize the channel capacity. A ZF receiver without precoding was suggested as the optimal transceiver architecture  for the optimally designed UCA system. For the non-optimal UCA system, a codebook-based precoder was proposed, in which the codebook is designed by quantization of the  center-shift angles and approximated power allocation. Simulation results showed that the ZF receiver achieves the channel capacity at the optimal design criteria, and when the optimal design criteria cannot be met, the proposed precoder can achieve the capacity with low feedback overhead. In future works, it would be interesting to extend this research to fast-moving scenarios, such as UAV backhaul systems or high-speed railway backhaul systems, which also have potential to leverage the robustness of misalignments of UCA-based MIMO systems over LoS channel environments.

\appendices
	
\section{Proof of Lemma 3}
\begin{figure*}[!t]
	\normalsize

\begin {equation}
\begin {split}
d(n,m)&=\{(R_t\cos{\theta_m'}-R_1\cos({\theta_n-\alpha_1}))^2+(R_t\sin{\theta_m'}-R_2\cos({\theta_n-\alpha_2})-D\sin\phi_{cs})^2+(R_3\cos({\theta_n-\alpha_3})+D\cos\phi_{cs})^2\}^{1\over2}\\
&=[R_t^2+R_r^2\{(b_{11}^2+b_{12}^2)\cos^2({\theta_n-\alpha_1})+(b_{21}^2+b_{22}^2)\cos^2({\theta_n-\alpha_2})+(b_{31}^2+b_{32}^2)\cos({\theta_n-\alpha_3})\}\\
&\quad -2R_tR_r\{\sqrt{b_{11}^2+b_{12}^2}\cos{\theta_m'}\cos({\theta_n-\alpha_1})+\sqrt{b_{21}^2+b_{22}^2}\sin{\theta_m'}\cos({\theta_n-\alpha_2})\}+2D\{R_r\sqrt{b_{21}^2+b_{22}^2}\cos({\theta_n-\alpha_2})\sin\phi_{cs}\\
&\quad+R_r\sqrt{b_{31}^2+b_{32}^2}\cos({\theta_n-\alpha_3})\cos\phi_{cs}-R_t\sin{\theta_m'}\sin\phi_{cs}+D^2\}]^{1\over2}\\
&\stackrel{(a)}{=}[R_t^2+R_r^2\{{1\over 2}(b_{11}^2+b_{12}^2+b_{21}^2+b_{22}^2+b_{31}^2+b_{32}^2)+{1\over 2}(\cos2({\theta_n-\alpha_1})+\cos2({\theta_n-\alpha_2})+\cos2({\theta_n-\alpha_3}))\}\\
&\quad -2R_tR_r\{\sqrt{b_{11}^2+b_{12}^2}\cos{\theta_m'}\cos({\theta_n-\alpha_1})+\sqrt{b_{21}^2+b_{22}^2}\sin{\theta_m'}\cos({\theta_n-\alpha_2})\}+2D\{R_r\sqrt{b_{21}^2+b_{22}^2}\cos({\theta_n-\alpha_2})\sin\phi_{cs}\\
&\quad+R_r\sqrt{b_{31}^2+b_{32}^2}\cos({\theta_n-\alpha_3})\cos\phi_{cs}-R_t\sin{\theta_m'}\sin\phi_{cs}+D^2\}]^{1\over2}\\
&\stackrel{(b)}{=}[R_t^2+{R_r^2 \over 2} \{\cos2({\theta_n-\alpha_1})+\cos2({\theta_n-\alpha_2})+\cos2({\theta_n-\alpha_3})\}\\
&\quad
\underbrace{-2R_tR_r\{\sqrt{b_{11}^2+b_{12}^2}\cos{\theta_m'}\cos({\theta_n-\alpha_1})+\sqrt{b_{21}^2+b_{22}^2}\sin{\theta_m'}\cos({\theta_n-\alpha_2})\}}_{\{1\}}\\
&\quad +2D\{R_r\sqrt{b_{21}^2+b_{22}^2}\cos({\theta_n-\alpha_2})\sin\phi_{cs}+R_r\sqrt{b_{31}^2+b_{32}^2}\cos({\theta_n-\alpha_3})\cos\phi_{cs}-R_t\sin{\theta_m'}\sin\phi_{cs}+D^2\}]^{1\over2}
\end {split}
\end {equation}
\begin{equation}
R_{\theta_{cs}}^{xy}R_{\varphi_x}^{xz}R_{\varphi_y}^{yz}=
\begin{bmatrix}
\cos{\theta_{cs}} \cos{\varphi_x}& -\cos{\theta_{cs}} \sin{\varphi_x} \sin{\varphi_y}-\sin{\theta_{cs}} \cos{\varphi_y} & -\cos{\theta_{cs}} \sin{\varphi_x} \cos{\varphi_y}+\sin{\theta_{cs}} \sin{\varphi_y}\\
\sin{\theta_{cs}} \cos{\varphi_x} & -\sin{\theta_{cs}} \sin{\varphi_x} \sin{\varphi_y}+\cos{\theta_{cs}} \cos{\varphi_y} & -\sin{\theta_{cs}} \sin{\varphi_x} \cos{\varphi_y}-\cos{\theta_{cs}} \sin{\varphi_y} \\
\sin{\varphi_x} & \cos{\varphi_x} \sin{\varphi_y} &  \cos{\varphi_x} \cos{\varphi_y}\\
\end{bmatrix} 
\end{equation}
\begin{equation}
\begin{split}
\{1\}&=-2R_tR_r\{{b_{11}}\cos{\theta_m'}\cos(\theta_n+\theta_o)+{b_{12}}\cos{\theta_m'}\sin(\theta_n+\theta_o)+{b_{21}}\sin{\theta_m'}\cos(\theta_n+\theta_o)+{b_{22}}\sin{\theta_m'}\sin(\theta_n+\theta_o)\}\\
&=-2R_tR_r\{\cos{\theta_{cs}}\cos{\varphi_x}\cos{\theta_m'}\cos({\theta_n+\theta_o})-\cos{\theta_{cs}}\sin{\varphi_x}\sin{\varphi_y}\cos{\theta_m'}\sin({\theta_n+\theta_o})\\
&\quad -\sin{\theta_{cs}}\cos{\varphi_y}\cos{\theta_m'}\sin({\theta_n+\theta_o})+\sin{\theta_{cs}}\cos{\varphi_x}\sin{\theta_m'}\cos({\theta_n+\theta_o})\\
&\quad -\sin{\theta_{cs}}\sin{\varphi_x}\sin{\varphi_y}\sin{\theta_m'}\sin({\theta_n+\theta_o})+\cos{\theta_{cs}}\cos{\varphi_y}\sin{\theta_m'}\sin({\theta_n+\theta_o})\}\\
&=-2R_tR_r\{\cos{\varphi_x}\cos({\theta_m'}-{\theta_{cs}})\cos({\theta_n+\theta_o})-\sin{\varphi_x}\sin{\varphi_y}\cos({\theta_m'}-{\theta_{cs}})\sin({\theta_n+\theta_o})+\cos{\varphi_y}\sin({\theta_m'}-{\theta_{cs}})\sin({\theta_n+\theta_o})\}\\
&=-2R_tR_r\{\cos({\theta_n}-{\theta_m'}+{\theta_o}+{\theta_{cs}})-2\sin^2{\varphi_x \over 2}\cos({\theta_m'}-{\theta_{cs}})\cos({\theta_n+\theta_o})-2\sin^2{\varphi_y \over 2}\sin({\theta_m'}-{\theta_{cs}})\sin({\theta_n+\theta_o})\\
&\quad-\sin{\varphi_x}\sin{\varphi_y}\cos({\theta_m'}-{\theta_{cs}})\sin({\theta_n+\theta_o})\}.
\end{split}
\end{equation}
\begin{equation}
\begin{split}
d(n,m)=&\Big[D^2+R_t^2+R_r^2-2R_tR_r\cos({\theta_n}-{\theta_m'}+{\theta_o}+{\theta_{cs}})
+{R_r^2\over 2}\big\{\cos2({\theta_n-\alpha_1})+\cos2({\theta_n-\alpha_2})+\cos2({\theta_n-\alpha_3})\big\}\\
&+2R_tR_r\big\{2\sin^2{\varphi_x \over 2}\cos({\theta_m'}-{\theta_{cs}})\cos({\theta_n+\theta_o})+2\sin^2{\varphi_y \over 2}\sin({\theta_m'}-{\theta_{cs}})\sin({\theta_n+\theta_o})\\
&+\sin{\varphi_x}\sin{\varphi_y}\cos({\theta_m'}-{\theta_{cs}})\sin({\theta_n+\theta_o})\big\}+2D\big\{R_2\cos({\theta_n-\alpha_2})\sin\phi_{cs}+R_3\cos({\theta_n-\alpha_3})\cos\phi_{cs}\big\}-R_t\sin{\theta_m'}\sin\phi_{cs}\Big]^{1\over2}\\
=&D\Big[1+{R_t^2+R_r^2 \over D^2}-{2R_tR_r \over D^2}\cos({\theta_n}-{\theta_m}+{\theta_o})
+{R_r^2\over 2D^2}\big\{(\cos2({\theta_n-\alpha_1})+\cos2({\theta_n-\alpha_2})+\cos2({\theta_n-\alpha_3}))\big\}\\
&+{4R_tR_r \over D^2}\big\{\sin^2{\varphi_x \over 2}\cos\theta_m\cos({\theta_n+\theta_o})+\sin^2{\varphi_y \over 2}\sin\theta_m\sin({\theta_n+\theta_o})+{1 \over 2}\sin{\varphi_x}\sin{\varphi_y}\cos\theta_m\sin({\theta_n+\theta_o})\big\}\\
&+{2 \over D}\big\{R_2\cos({\theta_n-\alpha_2})\sin\phi_{cs}+R_3\cos({\theta_n-\alpha_3})\cos\phi_{cs}-R_t\sin{\theta_m'}\sin\phi_{cs}\big\}\Big]^{1\over2}\\
=&D\big\{1+f(R_t,R_r,D,\theta_m,{\theta_n},{\theta_o},\theta_{cs},\varphi_x,\varphi_y,\phi_{cs})\big\}^{1\over2}
\end{split}
\end{equation}
\hrulefill
\end{figure*}
The distance $d(n,m)$ can be written as (20) where the equality (a) follows from $\cos^2{x}={\cos(2x)+1\over 2}$ and the equality (b) comes from the fact that $(b_{11}^2+b_{12}^2+b_{21}^2+b_{22}^2+b_{31}^2+b_{32}^2)=2$ and $b_{ij}$ is the ${(i,j)}^{th}$ element of $R_{\theta_{cs}}^{xy}R_{\varphi_x}^{xz}R_{\varphi_y}^{yz}$ given by (21). Using $\sqrt{b_{i1}^2+b_{i2}^2}\cos(\theta_n-\theta_o)={b_{i1}}\cos(\theta_n+\theta_o)+{b_{i2}}\sin(\theta_n+\theta_o)$,  the term \{1\} in (20) becomes (22) where the last equality follows from $\cos{x}=1-2\sin^2{x \over 2}$. From (22), $d(n,m)$ in (20) is rewritten as (23).

\section{Proofs of Properties 2(\textit{a})\textendash(\textit{e})}
\subsection{Proof of property 2\textit{(a)}}

Let us abbreviate $\sigma_{k}(\beta, \theta_o)$ to $\sigma_{k}$. Referring to (17), $\sigma_{k}$ is given by 
\begin{equation}
\begin{split}
\sigma_{k}
&=\bigg|\sum_{i=0}^{N_s-1}e^{-j\left({2\pi \over N_s}i(k-1)-\beta\cos{({2\pi \over N_s}i+\theta_o})\right)}\bigg|\\
&=\bigg|\sum_{i=0}^{{N_s\over2}-1}e^{-j\left({2\pi \over N_s}i(k-1)-\beta\cos{({2\pi \over N_s}i+\theta_o})\right)}\\
&\quad  +\sum_{i={N_s\over2}}^{{N_s}-1}e^{-j\left({2\pi \over N_s}i(k-1)-\beta\cos{({2\pi \over N_s}i+\theta_o})\right)}\bigg|\\
&=\bigg|\sum_{i=0}^{{N_s\over2}-1}e^{-j\left({2\pi \over N_s}i(k-1)\right)}\cdot e^{+j\left(\beta\cos{({2\pi \over N_s}i+\theta_o})\right)}\\
&\quad +\sum_{i=0}^{{N_s\over2}-1}e^{-j\left({2\pi \over N_s}i(k-1)+\pi(k-1)\right)}\cdot e^{-j\left(\beta\cos{({2\pi \over N_s}i+\theta_o})\right)}\bigg|\\
&=\bigg|\sum_{i=0}^{{N_s\over2}-1}e^{-j\left({2\pi \over N_s}i(k-1)\right)}\\
&\quad \times\Big\{e^{+j\left(\beta\cos{({2\pi \over N_s}i+\theta_o})\right)}+(-1)^{k-1}e^{-j\left(\beta\cos{({2\pi \over N_s}i+\theta_o})\right)}\Big\}\bigg|.\\
\end{split}
\end{equation}
Let $k'=N_s+2-k$, and $\sigma_{{k'}}$ is given by 
\begin{equation}
\begin{split}
\sigma_{{k'}}&=\bigg|\sum_{i=0}^{N_s-1}e^{-j\left({2\pi \over N_s}i(k'-1)-\beta\cos{({2\pi \over N_s}i+\theta_o})\right)}\bigg|\\
&=\bigg|\sum_{i=0}^{N_s-1}e^{+j\left({2\pi \over N_s}i(k-1)+\beta\cos{({2\pi \over N_s}i+\theta_o})\right)}\bigg|\\\nonumber
\end{split}
\end{equation}
\begin{equation}
\begin{split}
&=\bigg|\sum_{i=0}^{{N_s\over2}-1}e^{+j\left({2\pi \over N_s}i(k-1)+\beta\cos{({2\pi \over N_s}i+\theta_o})\right)}\\
&\quad+\sum_{i={N_s\over2}}^{{N_s}-1}e^{+j\left({2\pi \over N_s}i(k-1)+\beta\cos{({2\pi \over N_s}i+\theta_o})\right)}\bigg|\\\nonumber
\end{split}
\end{equation}
\begin{equation}
\begin{split}
&=\bigg|\sum_{i=0}^{{N_s\over2}-1}e^{+j\left({2\pi \over N_s}i(k-1)+\beta\cos{({2\pi \over N_s}i+\theta_o})\right)}\\
&\quad+\sum_{i=0}^{{N_s\over2}-1}e^{+j\left({2\pi \over N_s}i(k-1)+\pi(k-1)\right)}e^{-j\left(\beta\cos{({2\pi \over N_s}i+\theta_o})\right)}\bigg|\\
&=\bigg|\sum_{i=0}^{{N_s\over2}-1}e^{+j\left({2\pi \over N_s}i(k-1)\right)}\\
&\quad \times\Big\{e^{+j\left(\beta\cos{({2\pi \over N_s}i+\theta_o})\right)}+(-1)^{k-1}e^{-j\left(\beta\cos{({2\pi \over N_s}i+\theta_o})\right)}\Big\}\bigg|.\\
\end{split}
\end{equation}
To show that $\sigma_{{k}}=\sigma_{{k'}}$, we divide the proof into two cases for even and odd values of $k$. 
When $k$ is odd, $\sigma_{k}$ is given by
\begin{equation}
\begin{split}
\sigma_{k}&=2\bigg|\sum_{i=0}^{{N_s\over2}-1}e^{-j\left({2\pi \over N_s}i(k-1)\right)}\cos\left(\beta\cos{({2\pi \over N_s}i+\theta_o})\right)\bigg|\\
&=\bigg|\sum_{i=0}^{{N_s\over2}-1}e^{-j\left({2\pi \over N_s}i(k-1)\right)}\cos\left(\beta\cos{({2\pi \over N_s}i+\theta_o)}\right)\\
&+\sum_{i={N_s\over2}}^{N_s}e^{-j\left({2\pi \over N_s}(i-{N_s\over2})(k-1)\right)}\cos\left(\beta\cos{({2\pi \over N_s}(i-{N_s\over2})+\theta_o)}\right)\bigg|\\
&=\bigg|\sum_{i=0}^{N_s-1}e^{-j\left({2\pi \over N_s}i(k-1)\right)}\cos\left(\beta\cos{({2\pi \over N_s}i+\theta_o)}\right)\bigg|\\
&=\Bigg[\left\{\sum_{i=0}^{{N_s}-1}\cos\left({2\pi \over N_s}i(k-1)\right)\cos\left(\beta\cos{({2\pi \over N_s}i+\theta_o)}\right)\right\}^2\\
&\quad+\left\{\sum_{i=0}^{{N_s}-1}\sin\left({2\pi \over N_s}i(k-1)\right)\cos\left(\beta\cos{({2\pi \over N_s}i+\theta_o)}\right)\right\}^2\Bigg]^{1\over 2}.
\end{split}
\end{equation}
Similarly, $\sigma_{k'}$ is given by
\begin{equation}
\begin{split}
\sigma_{k'}&=2\bigg|\sum_{i=0}^{{N_s\over2}-1}e^{+j\left({2\pi \over N_s}i(k-1)\right)}\cos\left(\beta\cos{({2\pi \over N_s}i+\theta_o})\right)\bigg|\\
&=\Bigg[\left\{\sum_{i=0}^{{N_s}-1}\cos\left({2\pi \over N_s}i(k-1)\right)\cos\left(\beta\cos{({2\pi \over N_s}i+\theta_o)}\right)\right\}^2\\
&\quad+\left\{\sum_{i=0}^{{N_s}-1}\sin\left({2\pi \over N_s}i(k-1)\right)\cos\left(\beta\cos{({2\pi \over N_s}i+\theta_o)}\right)\right\}^2\Bigg]^{1\over 2}\\
&=\sigma_{k}.
\end{split}
\end{equation}
When $k$ is even, we have
\begin{equation}
\begin{split}
\sigma_{k}&=2\bigg|j\sum_{i=0}^{{N_s\over2}-1}e^{-j\left({2\pi \over N_s}i(k-1)\right)}\sin\left(\beta\cos{({2\pi \over N_s}i+\theta_o})\right)\bigg|\\
&=\bigg|j\sum_{i=0}^{{N_s}-1}e^{-j\left({2\pi \over N_s}i(k-1)\right)}\sin\left(\beta\cos{({2\pi \over N_s}i+\theta_o})\right)\bigg|
\end{split}
\end{equation}
and
\begin{equation}
\begin{split}
\sigma_{k'}&=2\bigg|j\sum_{i=0}^{{N_s\over2}-1}e^{+j\left({2\pi \over N_s}i(k-1)\right)}\sin\left(\beta\cos{({2\pi \over N_s}i+\theta_o})\right)\bigg|\\
&=\bigg|j\sum_{i=0}^{{N_s}-1}e^{+j\left({2\pi \over N_s}i(k-1)\right)}\sin\left(\beta\cos{({2\pi \over N_s}i+\theta_o})\right)\bigg|,
\end{split}
\end{equation}and the remaining proof is handled similarly to the odd case.

\subsection{Proof of property 2\textit{(b)}}
Let us consider odd and even values of $k$ separately. Referring to (26), when $k$ is odd,  $\sigma_{k}$ is upper bounded as 
\begin{equation}
\begin{split}
\sigma_{k}&=\bigg|\sum_{i=0}^{N_s-1}e^{-j\left({2\pi \over N_s}i(k-1)\right)}\cos\left(\beta\cos{({2\pi \over N_s}i+\theta_o)}\right)\bigg|\\
&\leq\sum_{i=0}^{N_s-1}\bigg|e^{-j\left({2\pi \over N_s}i(k-1)\right)}\cos\left(\beta\cos{({2\pi \over N_s}i+\theta_o)}\right)\bigg|\\
&=\sum_{i=0}^{N_s-1}\bigg|\cos\left(\beta\cos{({2\pi \over N_s}i+\theta_o)}\right)\bigg|,
\end{split}
\end{equation}
where the second inequality comes from the triangle inequality. What we want to find is a range of $\beta$ such that 
\begin{equation}
\sum_{i=0}^{N_s-1}\left|\cos\left(\beta\cos{({2\pi \over N_s}i+\theta_o)}\right)\right| \leq \left|\sum_{i=0}^{N_s-1}\cos\left(\beta\cos{({2\pi \over N_s}i+\theta_o)}\right)\right|,
\end{equation} where the upper bound is equal to $\sigma_{1}$. Because of the triangle inequality, (31) is true only when the equality holds. Thus, the range of $\beta$ is obtained as $0\leq \beta \leq \pi/2$ where $\cos\left(\beta\cos{({2\pi \over N_s}i+\theta_o)}\right)\geq 0$ for all $i$. \footnote{What we have interest is the range of $\beta$ starting from 0 to some value that we have $\sigma_{k}\leq \sigma_{1}$ for all $k\in \{2,\cdots, N_s\}$. Therefore, we ignore other valid ranges of $\beta$ which are over than $[0,\pi/2].$} Similarly,  $\sigma_{k}$ of the even value $k$ is given by
\begin{equation}
\begin{split}
\sigma_{k}&=2\bigg|j\sum_{i=0}^{{N_s\over2}-1}e^{-j\left({2\pi \over N_s}i(k-1)\right)}\sin\left(\beta\cos{({2\pi \over N_s}i+\theta_o})\right)\bigg|\\
&=\bigg|\sum_{i=0}^{N_s-1}je^{-j\left({2\pi \over N_s}i(k-1)\right)}\sin\left(\beta\cos{({2\pi \over N_s}i+\theta_o)}\right)\bigg|\\
&\leq\sum_{i=0}^{N_s-1}\bigg|je^{-j\left({2\pi \over N_s}i(k-1)\right)}\sin\left(\beta\cos{({2\pi \over N_s}i+\theta_o)}\right)\bigg|\\
&=\sum_{i=0}^{N_s-1}\bigg|\sin\left(\beta\cos{({2\pi \over N_s}i+\theta_o)}\right)\bigg|.\\
\end{split}
\end{equation}
To find the range of $\beta$ such that
\begin{equation}
\sum_{i=0}^{N_s-1}\bigg|\sin\left(\beta\cos{({2\pi \over N_s}i+\theta_o)}\right)\bigg| \leq \bigg|\sum_{i=0}^{N_s-1}\cos\left(\beta\cos{({2\pi \over N_s}i+\theta_o)}\right)\bigg|,
\end{equation} 
let $\theta_i \triangleq \beta\cos{({2\pi \over N_s}i+\theta_o)}$. To compare with the range of $\beta$ found in the case of odd value $k$, we consider $\beta \in [0, \pi/2]$. Then, (33) is rewritten as 
\begin{equation}
\sum_{i=0}^{N_s-1}|\sin\theta_i|\leq\sum_{i=0}^{N_s-1}|\cos\theta_i|,
\end{equation} where $-\pi/2\leq\theta_i\leq \pi/2$ and $\cos\theta_i\geq0$. When $0\leq \beta \leq \pi/4$, (34) is always true because $-\pi/4\leq\theta_i\leq\pi/4$; thus, $|\sin\theta_i|\leq\cos\theta_i$ for all $\theta_i$. When $ \pi/4\leq \beta \leq \pi/2$, $-\pi/2\leq\theta_i\leq\pi/2$, let $\theta_{i,1} \in \{\theta_i|-\pi/2<{2\pi \over N_s}i+\theta_o\leq \pi/2\}$ and $\theta_{i,2} \in \{\theta_i|-\pi<{2\pi \over N_s}i+\theta_o\leq -\pi/2 \quad or \quad \pi/2<{2\pi \over N_s}i+\theta_o \leq \pi\}$, respectively, and the numbers of $\theta_{i,1}$ and $\theta_{i,2}$ are $N_s/2$. Because  $|\sin\theta_i|$ is concave in each ranges, we have $\mathbb{E}(\left|\sin\theta_{i,1}\right|)+\mathbb{E}(\left|\sin{\theta_{i,2}}\right|)\leq\sin\mathbb{E}(\left|\theta_{i,1}\right|)+\sin\mathbb{E}(\left|\theta_{i,2}\right|)$. Furthermore, because $0\leq\mathbb{E}(\left|\theta_{i,1}\right|), \mathbb{E}(\left|\theta_{i,2}\right|)\leq\pi/2$, we have $\sin\mathbb{E}(\left|\theta_{i,1}\right|)+\sin\mathbb{E}(\left|\theta_{i,2}\right|)\leq 2\sin\mathbb{E}(\left|\theta_i\right|)$. Thus,
\begin{equation}
\begin{split}
\sum_{i=0}^{N_s-1}|\sin\theta_i|&\leq {N_s\over 2}(\sin\mathbb{E}(\left|\theta_{i,1}\right|)+\sin\mathbb{E}(\left|\theta_{i,2}\right|))\\
&\leq  N_s\sin\mathbb{E}(\left|\theta_i\right|).
\end{split}
\end{equation}
Similarly, $\sum_{i=0}^{N_s-1}|\cos\theta_i|\leq N_s\cos\mathbb{E}(\left|\theta_i\right|)$.
Using this, the condition (34) can be rewritten as
\begin{equation}
\begin{split}
\sin\mathbb{E}(\left|\theta_i\right|)&\leq\cos\mathbb{E}(\left|\theta_i\right|).
\end{split}
\end{equation}
To hold (36),  $0\leq\mathbb{E}(|\theta_i|) \leq {\pi \over 4}$ and the range of $\beta$ is obtained as $\pi/4\leq\beta\leq {\pi N_s \over 4\sum_{i=0}^{N_s-1} |\cos{({2\pi \over N_s}i+\theta_o})|}$. Combining all results, the objective range of $\beta$ is given by
\begin{equation}
0\leq\beta\leq {\pi N_s \over 4\sum_{i=0}^{N_s-1} |\cos{({2\pi \over N_s}i+\theta_o})|}.
\end{equation}
\subsection{Proof of property 2\textit{(c)}}
Let us assume that $\theta_o$ is fixed at some value. Then, $\sigma_{k}(\beta, \theta_o)$ can be written as $\sigma_{k}(\beta)$. Because $\sigma_{k}(\beta)$ is a composition function of  a sum of continuous sinusoidal functions and an absolute value function which are continuous functions, it is continuous. Therefore, if $\beta\rightarrow 0$ then $\sigma_{k}(\beta)\rightarrow \sigma_{k}(0)$. 
When $\beta = 0$, we have
\begin{equation}
\begin{split}
\sigma_{k}(0)&=\bigg|\sum_{i=0}^{N_s-1}e^{-j\{{2\pi \over N_s}i(k-1)\}}\bigg|\\
\end{split}
\end{equation}
and it can be easily seen that $\sigma_{k}(0)=N_s$ when $k=1$. Otherwise, the sum in (38) is a finite geometric series which is given by
\begin{equation}
\begin{split}
\sigma_{k}(0)&=\bigg|{ {1-[e^{-j2\pi(k-1)/N_s}]^{N_s}}\over{1-{ e^{-j2\pi(k-1)/N_s}}}}\bigg|\\
&=\bigg|{ {1-e^{-j2\pi(k-1)}}\over{1-{ e^{-j2\pi(k-1)/N_s}}}}\bigg|.
\end{split}
\end{equation}
Therefore, $\sigma_{k}(0) = 0$ for $k=2,\cdots, N_s$.

\subsection{Proof of property 2\textit{(d)}}

Let $i={N_s}-1-i'$ where ${i, i'\in \{0,\cdots, N_s-1\}}$. Referring to (26), when $k$ is odd, $\sigma_{k}(\beta, -\theta_o)$ is given by
\begin{equation}
\begin{split}
\sigma_{k}(\beta,&-\theta_o)=\bigg|\sum_{i=0}^{N_s-1}e^{-j\left({2\pi \over N_s}i(k-1)\right)}\cos\left(\beta\cos{({2\pi \over N_s}i-\theta_o)}\right)\bigg|\\
&=\bigg|\sum_{i'=0}^{N_s-1}e^{-j\left({2\pi \over N_s}({N_s}-1-i')(k-1)\right)}\\
&\quad\times \cos\left(\beta\cos{({2\pi \over N_s}({N_s}-1-i')-\theta_o)}\right)\bigg|\\
&=\bigg|\sum_{i'=0}^{N_s-1}e^{+j\left({2\pi \over N_s}(i'+1)(k-1)\right)}\cos\left(\beta\cos{({2\pi \over N_s}(i'+1)+\theta_o)}\right)\bigg|\\
&\stackrel{(a)}{=}\left|\sum_{i'=0}^{N_s-1}e^{+j\left({2\pi \over N_s}i'(k-1)\right)}\cos\left(\beta\cos{({2\pi \over N_s}i'+\theta_o)}\right)\right|\\
&=\bigg|\sum_{i'=0}^{{N_s}-1}\cos\left({2\pi \over N_s}i'(k-1)\right)\cos\left(\beta\cos{({2\pi \over N_s}i'+\theta_o)}\right)\\
&+j\sum_{i'=0}^{{N_s}-1}\sin\left({2\pi \over N_s}i'(k-1)\right)\cos\left(\beta\cos{({2\pi \over N_s}i'+\theta_o)}\right)\bigg|\\
&=\Bigg[\bigg\{\sum_{i'=0}^{{N_s}-1}\cos\left({2\pi \over N_s}i'(k-1)\right)\cos\left(\beta\cos{({2\pi \over N_s}i'+\theta_o)}\right)\bigg\}^2\\
&+\bigg\{\sum_{i'=0}^{{N_s}-1}\sin\left({2\pi \over N_s}i'(k-1)\right)\cos\left(\beta\cos{({2\pi \over N_s}i'+\theta_o)}\right)\bigg\}^2\Bigg]^{1\over 2}\\
&=\sigma_{k}(\beta,\theta_o)
\end{split}
\end{equation} where the equality (a) comes from the fact that
\begin{equation}\begin{split} &e^{-j\left({2\pi \over N_s}(N_s)(k-1)\right)}\cos\left(\beta\cos{({2\pi \over N_s}(N_s)-\theta_o)}\right)\\&=e^{-j\left({2\pi \over N_s}(0)(k-1)\right)}\cos\left(\beta\cos{({2\pi \over N_s}(0)-\theta_o)}\right).\nonumber\end{split}\end{equation}
Similarly, when $k$ is even, we get
\begin{equation}
\begin{split}
\sigma_{k}(\beta,&-\theta_o)=\bigg|j\sum_{i'=0}^{N_s-1}e^{+j\left({2\pi \over N_s}i'(k-1)\right)}\sin\left(\beta\cos{({2\pi \over N_s}i'+\theta_o)}\right)\bigg|\\
&=\bigg|j\sum_{i'=0}^{{N_s}-1}\cos\left({2\pi \over N_s}i'(k-1)\right)\sin\left(\beta\cos{({2\pi \over N_s}i'+\theta_o)}\right)\\
&-\sum_{i'=0}^{{N_s}-1}\sin\left({2\pi \over N_s}i'(k-1)\right)\sin\left(\beta\cos{({2\pi \over N_s}i'+\theta_o)}\right)\bigg|\\
&=\Bigg[\left\{\sum_{i'=0}^{{N_s}-1}\cos\left({2\pi \over N_s}i'(k-1)\right)\sin\left(\beta\cos{({2\pi \over N_s}i'+\theta_o)}\right)\right\}^2\\
&+\left\{\sum_{i'=0}^{{N_s}-1}\sin\left({2\pi \over N_s}i'(k-1)\right)\sin\left(\beta\cos{({2\pi \over N_s}i'+\theta_o)}\right)\right\}^2\Bigg]^{1\over 2}\\
&=\sigma_{k}(\beta,\theta_o).
\end{split}
\end{equation}
\subsection{Proof of property 2\textit{(e)}}
When $k=N_s/2+1$, $\sigma_{k}(\beta,\theta_o)$ is decomposed into odd and even values of $i$ as
\begin{equation}
\begin{split}
\sigma_{{N_s/2+1}}(\beta,\theta_o)&=\bigg|\sum_{i=0}^{N_s-1}e^{-j{\pi}i}e^{+j\left(\beta\cos{({2\pi \over N_s}i+\theta_o})\right)}\bigg|\\
&=\bigg|\sum_{i'=0}^{N_s/2-1}\underbrace{e^{+j\left(\beta\cos{({2\pi \over N_s}(2i')+\theta_o})\right)}}_{even \ i}\\
&\quad-\underbrace{e^{+j\left(\beta\cos{({2\pi \over N_s}(2i'+1)+\theta_o})\right)}}_{odd \ i}\bigg|.
\end{split}
\end{equation}
Let $i'={N_s\over2}-1-i''$, then the term of odd value $i$ in (42) can be written as
\begin{equation}
\begin{split}
\cos{({2\pi \over N_s}(2i'+1)+{\theta_o}})&=\cos{({2\pi \over N_s}(2({N_s\over2}-1-i'')+1)+{\theta_o}})\\
&=\cos{(2\pi-{2\pi \over N_s}(2i''+1)+{\theta_o}})\\
&=\cos{({2\pi \over N_s}(2i''+1)-{\theta_o}}).\\
\end{split}
\end{equation}
Using this, (42) can be rewritten as 
\begin{equation}
\begin{split}
\sigma_{{N_s/2+1}}(\beta,\theta_o)&=\bigg|\sum_{i'=0}^{N_s/2-1}e^{+j\left(\beta\cos{({2\pi \over N_s}(2i')+\theta_o})\right)}\\
&-\sum_{i''=0}^{N_s/2-1}e^{+j\left(\beta\cos{({2\pi \over N_s}(2i''+1)-\theta_o})\right)}\bigg|.
\end{split}
\end{equation}
Therefore, when $\theta_o=\pi/N_s$, 
\begin{equation}
\begin{split}
\sigma_{c_{N_s/2+1}}(\beta,\pi/N_s)&=\bigg|\sum_{i'=0}^{N_s/2-1}e^{+j\left(\beta\cos{({2\pi \over N_s}(2i')+{\pi \over N_s}})\right)}\\
&-\sum_{i''=0}^{N_s/2-1}e^{+j\left(\beta\cos{({2\pi \over N_s}(2i'')+{\pi \over N_s}})\right)}\bigg|\\
&=0
\end{split}
\end{equation}
and the same result can be achieved when $\theta_o=-\pi/N_s$ because of the property 2\textit{(d)}.



\ifCLASSOPTIONcaptionsoff
  \newpage
\fi


\begin{thebibliography}{1}
	\bibitem{IEEEhowto:kopka} P. Wang, Y. Li, X. Yuan, L. Song, and B. Vucetic, ``Tens of gigabits wireless communications over E-band LoS MIMO channels with uniform linear antenna arrays,'' \emph{IEEE Trans. Wireless Commun}., vol. 13, no. 7, pp. 3791–3805, Jul. 2014.
	\bibitem{IEEEhowto:kopka} Xiaohang Song, Wolfgang Rave, Nithin Babu, Sudhan Majhi and Gerhard Fettweis, ``Two-level spatial multiplexing using hybrid beamforming for millimeter-wave backhaul," \emph{IEEE Trans. Wireless Commun}., vol. 17, no. 7, pp. 4830-4844, July. 2018.
	
	
	\bibitem{IEEEhowto:kopka} Eric Torkildson, Upamanyu Madhow, and Mark Rodwell, ``Indoor millimeter
	wave MIMO: feasibility and performance,'' \emph{IEEE Trans. Wireless Commun}., vol. 10, no. 12, pp. 4150-4160, Dec. 2011.
	
	\bibitem{IEEEhowto:kopka} Xiaohang Song and Gerhard Fettweis, ``On spatial multiplexing of strong line-of-sight
	MIMO with 3D antenna arrangements,'' \emph{IEEE Commun. Lett.}, vol. 4, no. 4, pp. 393-396, Aug. 2015.
	\bibitem{IEEEhowto:kopka} F. Bohagen, P. Orten, and G. Oien, ``Design of optimal high-rank line-of-sight MIMO channels,'' \emph {IEEE Trans. Wireless Commun.}, vol. 6, no. 4, pp. 1420–1425, Apr. 2007.
	\bibitem{IEEEhowto:kopka} I. Sarris and A. Nix, ``Design and performance assessment of high-capacity MIMO
	architectures in the presence of a line-of-sight component,'' \emph {IEEE Trans. Veh. Technol.}, vol. 56, no. 4, pp. 2194-2202, Jul. 2007.
	
	
	\bibitem{IEEEhowto:kopka} Lakshmi Natarajan, Yi Hong, Senior Member, IEEE, and Emanuele Viterbo, ``Line-of-sight $2\times nr$ MIMO with random antenna orientations,'' \emph{IEEE Trans. Veh. Technol}., vol. 66, no. 6, pp. 5134-5147, Jun. 2017.
	\bibitem{IEEEhowto:kopka} Xumin pu, Shihai Shao, Kai Deng and Youxi Tang, ``Effects of array orientations on degrees of freedom for 3D LoS channels in short-range communications,'' \emph{IEEE Commun. Lett.}, vol. 4, no. 1, pp. 106-109, Feb. 2015.
	
	
	
	\bibitem{IEEEhowto:kopka} Frode Bohagen, Pal Orte, and Geir E. Oien, ``Optimal design of uniform planar antenna arrays for strong line-of-sight MIMO channels,'' in \emph{Proc. IEEE 7th Workshop on Signal Processing Advances in Wireless Communications}, Cannes, France, Jul. 2006.
	
	\bibitem{IEEEhowto:kopka} P. Wang, Y. Li and B. Vucetic, ``Millimeter wave communications with symmetric uniform circular antenna arrays,'' \emph{IEEE Commun. Lett}., vol. 18, no. 8, pp. 1307-1310, Aug. 2014.
	\bibitem{IEEEhowto:kopka} Liang Zhou and Yoji Ohashi, ``Performance analysis of mmWave LOS-MIMO system with Uniform Circular Array,'' in  \emph{Proc. IEEE VTC 2015-Spring}, May. 2015.
	
	
	\bibitem{IEEEhowto:kopka} O. Edfors and A. J. Johansson, ``Is orbital angular momentum (OAM) based radio communication an unexploited area?'' \emph{IEEE Trans. Antennas
		and Propag}., vol. 60, no. 2, pp. 1126-1131, Feb. 2012.
	\bibitem{IEEEhowto:kopka} Li Zhu and Jiang Zhu, ``Optimal design of uniform circular antenna array in mmWave LoS MIMO channel," \emph{IEEE Access}, vol. 6, pp. 61022 - 61029, Sep. 2018.
	
	
	\bibitem{IEEEhowto:kopka} Haiyue Jing, Wenchi Cheng, and Xiang-Gen Xia, ``A simple channel independent beamforming scheme with parallel uniform circular array," \emph{IEEE Communications Letters}, vol. 23, no. 3, pp. 414-417, Mar. 2019. 
	
	
	
	\bibitem{IEEEhowto:kopka} Yuri Jeon, Minhyun Kim, Gye-Tae Gil and Yong H. Lee, ``LoS spatial multiplexing and beamforming using uniform circular array of subarrays," in \emph{Proc. IEEE VTC 2016-Spring}, May. 2016.
	
	\bibitem{IEEEhowto:kopka} W. Cheng, W. Zhang, H, Jing, S. Gao, and H. Zhang, ''Orbital angular momentum for wireless communications,'' \emph{IEEE Wireless Commun}.,
	vol. 26, no. 1, pp. 100–107, Feb. 2019.
	
	\bibitem{IEEEhowto:kopka} W. Zhang, S. Zheng, X. Hui, R. Dong, X. Jin, H. Chi, and X. Zhang, ''Mode Division Multiplexing Communication Using Microwave Orbital Angular Momentum: An Experimental Study,'' \emph{IEEE Trans. Wireless Commun}., vol. 16, no. 2, pp. 1308–1318, Feb. 2017.
	
	
	\bibitem{IEEEhowto:kopka} W. Cheng, H. Zhang, L. Liang, H. Jing, and Z. Li, ''Orbital-Angular-Momentum Embedded Massive MIMO: Achieving Multiplicative Spectrum-Efficiency for mmWave Communications,'' \emph{IEEE Access}.,
	vol. 6, pp. 2732–2745, Dec. 2017.
	
	\bibitem{IEEEhowto:kopka} Rui Chen, Hui Xu, Marco Moretti, and Jiandong Li, ''Beam Steering for the Misalignment in UCA-Based
	OAM Communication Systems,'' \emph{IEEE Wireless Commun. Lett}.,
	vol. 7, no. 4, pp. 582-585, Aug. 2018.
	
	\bibitem{IEEEhowto:kopka} David Tse and Pramod Viswanath,  \textit{Fundamentals of Wireless Communication}, New York: Cambridge University Press, 2005. 
	
	
	
	\bibitem{IEEEhowto:kopka} P. Ioannides and C.A. Balanis, ''Uniform circular arrays for smart antennas,'' \emph{ IEEE Antennas and Propagation Magazine}.,
	vol. 47, Issue. 4, pp. 192-206, Aug. 2005.
	
\end{thebibliography}
\end{document}